%% file: master.tex
\documentclass[fleqn,12pt,twoside]{article}

\usepackage{graphicx}
\input{extern/layout_commands}
\input{extern/math_commands}
\input{extern/draft_commands}
\input{my_commands}

\title{\titlefamily\Huge 
The area of horizons and the trapped region
}
\author{\titlefamily Lars Andersson
  \thanks{Supported in part by the NSF, under contract no. DMS 0407732 with the University of Miami.}\\
  \titlefamily\small lars.andersson@aei.mpg.de\\
  \titlefamily\small Albert-Einstein-Institut,
  \titlefamily\small Am M\"uhlenberg 1, 14476 Potsdam, Germany.\\[0.5ex]
  \titlefamily\small Department of Mathematics, University of Miami,
  \titlefamily\small Coral Gables, FL 33124, USA.\\  
  \and
  \titlefamily Jan Metzger
  \thanks{Supported in part by a Feodor-Lynen Fellowship of the Humboldt Foundation.}\\
  \titlefamily\small jan.metzger@aei.mpg.de\\
  \titlefamily\small Albert-Einstein-Institut, 
  \titlefamily\small Am M\"uhlenberg 1, D-14476 Potsdam,
  Germany.\\[0.5ex]
  \titlefamily\small Stanford University, Mathematics,
  450 Serra Mall, Stanford, CA 94305, USA.
}
\date{}

\begin{document}
\hyphenation{}
\pagestyle{footnumber}
\maketitle
\thispagestyle{footnumber}
\begin{abst}%
  This paper considers some fundamental questions concerning
  marginally trapped surfaces, or apparent horizons, in Cauchy data
  sets for the Einstein equation.  An area estimate for outermost
  marginally trapped surfaces is proved.  The proof makes use of an
  existence result for marginal surfaces, in the presence of barriers,
  curvature estimates, together with a novel surgery construction for
  marginal surfaces.  These results are applied to characterize the
  boundary of the trapped region.
\end{abst}
\input{intro}
% 
\input{prelim}
%
\input{existence}
%
\input{stable}
%
\input{improved}
% 
\input{surgery}

%  
\input{trapped}
% 
%\input{conclusion}
%
\input{ack}

%
\bibliographystyle{amsalpha}
\bibliography{extern/references}
\end{document}

%% file: extern/math_commands.tex
\usepackage{amsmath}
\usepackage{amssymb}

\newcounter{mynumcounter}
\newenvironment{mynum}%
	       {%
		 \begin{list}{(\roman{mynumcounter})\hspace*{\fill}}%
		   {
		     \setlength{\topsep}{0cm}
		     \setlength{\partopsep}{0cm}
		     \setlength{\itemsep}{0ex}
		     \setlength{\parsep}{0cm}
		     \setlength{\leftmargin}{0cm}
		     \setlength{\itemindent}{8mm}		   
		     \setlength{\labelsep}{3mm}
		     \setlength{\labelwidth}{5mm}
		     \usecounter{mynumcounter}
		   }%
	       }
	       {\end{list}}

\newcommand{\eps}{\varepsilon}
\newcommand{\del}{\partial}

\newcommand{\dd}[2]{\frac{\del #1}{\del #2}}
\newcommand{\ddeval}[3]{\left.\dd{#1}{#2}\right|_{#3}}

\newcommand{\DD}[2]{\frac{d #1}{d #2}}

\newcommand{\Vol}{\operatorname{Vol}}

\newcommand{\graph}{\operatorname{graph}}

\newcommand{\tr}{\operatorname{tr}}
\renewcommand{\div}{\operatorname{div}}

\newcommand{\IR}{\mathbf{R}}

\newcommand{\CH}{\mathcal{H}}

\newcommand{\CJ}{\mathcal{J}}

\newcommand{\CP}{\mathcal{P}}

\newcommand{\CT}{\mathcal{T}}

\newcommand{\id}{\operatorname{id}}

\newcommand{\rmd}{\mathrm{d}}

\newcommand{\dist}{\mathrm{dist}}

\newcommand{\dmu}{\,\rmd\mu}

\newcommand{\Dist}[1]{\smash{\sideset{^{#1}}{}{\mathop\mathrm{dist}\nolimits}}}

\newcommand{\distM}{\!\Dist{M}}
\newcommand{\distSig}{\!\Dist{\Sigma}}

\newcommand{\ra}{\rangle}
\newcommand{\la}{\langle}
\newcommand{\inj}{\mathrm{inj}}

\newcommand{\lap}{\Delta}
\newcommand{\Laplacian}[1]{\smash{\sideset{^{#1}}{}{\mathop\lap\nolimits}}}

\newcommand{\lapSig}{\Laplacian{\Sigma}}

\newcommand{\Connection}[1]{\smash{\sideset{^{#1}}{}{\mathop\nabla\nolimits}}}

\newcommand{\nabSig}{\!\Connection{\Sigma}}

\newcommand{\Divergence}[1]{\smash{\sideset{^{#1}}{}{\mathop\mathrm{div}\nolimits}}}

\newcommand{\divSig}{\!\Divergence{\Sigma}}

\newcommand{\divM}{\!\Divergence{M}}

\newcommand{\Trace}[1]{\smash{\sideset{^{#1}}{}{\tr}}}

\newcommand{\trSig}{\Trace{\Sigma}}
\newcommand{\trM}{\Trace{M}}

\newcommand{\Riemann}[1]{\smash{\sideset{^{#1}}{}{\mathop\mathrm{Rm}\nolimits}}}

\newcommand{\RiemM}{\!\Riemann{M}}

\newcommand{\Scalarcurv}[1]{\smash{\sideset{^{#1}}{}{\mathop\mathrm{Sc}\nolimits}}}

\newcommand{\ScalM}{\!\Scalarcurv{M}}
\newcommand{\ScalSig}{\!\Scalarcurv{\Sigma}}

\newcommand{\Acircbar}%
{\hspace*{4pt}\raisebox{8.5pt}{\makebox[-4pt][l]{$\scriptstyle\circ$}}\bar A}
\newcommand{\Kcircbar}%
{\hspace*{4pt}\raisebox{8.5pt}{\makebox[-4pt][l]{$\scriptstyle\circ$}}\bar K}

\newcommand{\half}{\tfrac{1}{2}}

%% file: my_commands.tex
\newcommand{\mnote}[1]{}     %no marginal notes

\newcommand{\two}{\mathrm{I\!I}}

\newcommand{\todef}[1]{\emph{#1}}

%% file: intro.tex
\section{Introduction}
\label{sec:introduction}
Trapped and marginally trapped surfaces play a central role in the
analysis of spacetime geometry. By the singularity theorems of Hawking
and Penrose \cite{Hawking-Ellis:1973}, a spacetime which satisfies
suitable energy and causality conditions, and which in addition
contains a trapped surface, must contain a black hole.  Marginally
trapped surfaces, or apparent horizons, serve as the quasi-local
version of black hole boundary. In numerical general relativity, they
are used as excision surfaces for the evolution of black hole initial
data, and approximations to physical characteristics of a black
hole such as linear and angular momentum
\cite{Krishnan:2007pu,Campanelli:2006fy} can be calculated in terms of
data induced on the apparent horizon.

We briefly recall some basic facts. 
A two dimensional spacelike surface $\Sigma$ 
in a 4-dimensional Lorentzian spacetime
has, up to normalization, 
two future pointing null normals. We designate one of these, $\ell^+$, the
outward pointing, and the other $\ell^-$, the inward pointing null normal. 
Corresponding to $\ell^\pm$ we have the null mean 
curvatures or null expansions $\theta^{\pm}$. Let $(M,g,K)$
be a Cauchy data set containing $\Sigma$. Then $\theta^{\pm}$ is given
by
\begin{equation*}
  \theta^{\pm} = P \pm H  
\end{equation*}
where $H$ is the mean curvature of $\Sigma$ in $M$ with respect to the
outward pointing normal, and $P = \tr^\Sigma K$, the trace of the
projection of $K$ to $\Sigma$. The surface $\Sigma$ is said to be
(future) trapped if $\theta^\pm < 0$, and (future) marginally trapped
if $\theta^- < 0$, while $\theta^+ = 0$.  If $\theta^+ < 0$ or
$\theta^+ > 0$, with no condition imposed on $\theta^-$, then $\Sigma$
is called outer trapped or outer untrapped, respectively. Finally, if
the condition $\theta^+ = 0$ holds, with no further condition on
$\theta^-$, then $\Sigma$ is called a marginally outer trapped
surface, or MOTS. We will explicitly review notation and further
conditions needed on $(M,g,K)$ in section~\ref{sec:preliminaries}.

From a mathematical point of view, MOTS are the natural generalization
of minimal surfaces to a Lorentzian setting, see the discussion in
\cite{Andersson-Metzger:2005}. In particular, in the case of
time-symmetric Cauchy data, where $K \equiv 0$, a MOTS is a minimal
surface. However, a fundamental difference between minimal surfaces
and MOTS, is that MOTS are not stationary with respect to an elliptic
functional.
In spite of this, there
is a notion of stability for MOTS 
analogous to the notion of stability for minimal
surfaces, cf. \cite{Andersson-Mars-Simon:2005,andersson-mars-simon:2007}.
Although the 
stability operator in the case of MOTS fails to be
self-adjoint, many of the results and ideas generalize from the case of
stable minimal surfaces to  the case of stable MOTS. In particular, a
curvature estimate, 
generalizing the classical result of
\cite{Schoen-Simon-Yau:1975}  was proved in
\cite{Andersson-Metzger:2005} for the case of stable MOTS. 

The so-called Jang's equation \cite{Jang:1978} is closely related to
the equation $\theta^+ = 0$. Both are prescribed mean curvature
equations, where the right hand side depend on the normal. A careful
study of Jang's equation is a crucial ingredient in the positive mass
proof of Schoen and Yau \cite{Schoen-Yau:1981}.  Among other things,
their argument makes use of the fact that the boundary of the blowup
set for Jang's equation consists of marginal surfaces. This means that
the question of existence of MOTS may be approached by studying the
existence of blowup solutions to Jang's equation. This observation was
used by Yau \cite{Yau:2001} to give a criterion for a Cauchy data set
to contain a marginal surface.

A consequence of the fact that MOTS are not critical points for a
variational principle is that the familiar barrier arguments for the
existence of minimal surfaces do not generalize to MOTS.  However, as
was pointed out by Schoen in a talk given at the Miami Waves
conference in 2004 \cite{Schoen:2004}, the fact that blowup surfaces
for Jang's equation are marginal surfaces actually provides a result
which replaces the above mentioned barrier arguments.
\begin{theorem}
  \label{thm:schoen-intro}
  Let $(M,g,K)$ be a Cauchy data set. Assume that $M$ is compact with
  two boundary components, an inner and an outer boundary and assume
  that the inner boundary is outer trapped and the outer boundary is
  outer untrapped. Then $M$ contains a stable MOTS.
\end{theorem}
This theorem is a consequence of Schoen's original result, stated as
theorem~\ref{thm:schoen} and a closer analysis of the blow-up surface, cf. 
theorem~\ref{thm:schoen-stable}. Unfortunately, a proof of
theorem~\ref{thm:schoen} has not been published. In
section~\ref{sec:theorem-schoen} we therefore provide a detailed proof
of this result, of which we will make use of throughout the present
paper.

We wish to remark here that if the ambient manifold is asymptotically
flat with appropriate fall-off conditions, then spheres near infinity
will be untrapped and can serve as outer barriers in
theorem~\ref{thm:schoen-intro}.

Starting from the curvature estimates for MOTS mentioned above, it is
easy to show that the set of all stable marginally trapped surfaces in
a compact region is compact, given a uniform estimate for the area.
However, such an estimate cannot be expect to hold in general.
Examples due to Colding-Minicozzi and others
\cite{Colding-Minicozzi:1999,Dean:2003} show that for each genus
$g\geq 1$ there is an example of a compact three dimensional manifold containing a
sequence of stable minimal surfaces of genus $g$ with unbounded area.
Recalling that minimal surfaces are MOTS in the special case $K=0$,
this shows that an a priori area estimate for MOTS requires further
conditions.

If we consider surfaces \emph{minimizing} area in a given homology
class, on the other hand, there is no need to prove an area bound to
obtain compactness, as one can assume that the area is bounded by the
area of any comparison surface. For the case of MOTS, the appropriate
analogue of a minimizing surface is an \emph{outermost} MOTS. We say
that a MOTS $\Sigma$ is outermost in $M$ if there is no other MOTS in
the complement of the region which $\Sigma$ bounds with a, possibly
empty, inner boundary. In this respect, the main result of this paper,
cf.\ theorem~\ref{thm:area-bound} is an area estimate for the
outermost MOTS.
\begin{theorem}
  \label{thm:intro-bound}
  There exists a constant $C$ which is an increasing function of
  $\|\RiemM\|_{C^0(M)}$, $\|K\|_{C^1(M)}$, $\inj_\rho(M,g,K;\del
  M)^{-1}$, and $\Vol M$ such that the area of an outermost MOTS
  $\Sigma$ satisfies the estimate
  \begin{equation*}
    |\Sigma|\leq C.
  \end{equation*}
  The quantity $\inj_\rho(M,g,K;\del M)$ is explained in
  definition~\ref{def:injectivity_radius}.
\end{theorem}
This result does not require the MOTS to be connected.  Thus, in
combination with the curvature estimate for stable MOTS we infer an
estimate for the number of components of the outermost MOTS.

Note, even for outward minimizing surfaces the above bound does not
actually follow from the variational principle, as it does not refer
to the area of a comparison surface. In this respect our area estimate
is related to the area estimate in \cite{nabutowsky:rotman:2006} for
minimizing minimal surfaces in terms of volume and the homologial
filling functions of the ambient manifold, which must have simple
enough homology.

To put theorem~\ref{thm:intro-bound} into perspective, recall that the
Penrose inequality is a conjectured relation between the ADM mass and
the area of the horizon. For a general Cauchy data set, the exact
statement of the Penrose inequality is a subtle issue. Although, the
area estimate stated in theorem~\ref{thm:intro-bound} holds for
outermost MOTS, a counter example due to Ben-Dov~\cite{Ben-Dov:2004}
shows that an inequality between the area of the outermost MOTS and
the ADM mass does not hold in general.

One of the main steps in the proof of theorem~\ref{thm:intro-bound} is
a surgery argument, which is given in section~\ref{sec:surgery}. This
argument constructs, given a stable MOTS $\Sigma$ with sufficiently
large area and an outer barrier surface, another stable MOTS outside
$\Sigma$. The two main steps in the argument is to show, using the
curvature estimate, that given a stable MOTS with sufficiently large
area, it is possible to glue in a neck with negative $\theta^+$,
thereby constructing a $\Sigma'$ outside $\Sigma$ with $\theta^+ \leq
0$. Together with theorem~\ref{thm:schoen-intro} this yields a contradiction to the
assumption that $\Sigma$ is outermost.

The surgery argument may also be used to give a replacement for the
strong maximum principle for outermost MOTS. It should be noted that
for general MOTS, the strong maximum principle does not apply in
general, in particular it can not be used to rule out that a surface
touches itself in points where the the normals of the two touching
pieces point into opposite directions. This is the exactly the
situation which we can address with the surgery argument.

Combining the above area estimate for outermost MOTS and the
curvature estimate of \cite{Andersson-Metzger:2005} yields, as already
mentioned, a compactness result for the class of outermost MOTS 
in a compact region. Using this fact in combination with the surgery
technique discussed above enables us to give a characterization of the
boundary of the trapped region. 

The outer trapped region is the union of all domains bounded by a
weakly outer trapped surface and the, possibly empty, interior
boundary of the initial data set. It has been proposed by several
authors that the boundary of the outer trapped region is a smooth
MOTS. However, the arguments put forth to prove this, see for example
\cite{Hawking-Ellis:1973,Kriele-Hayward:1997}, relied on strong extra
assumptions such as a piecewise smoothness of the boundary. Using the
techniques developed in this paper we are able to settle this problem
completely.
\begin{theorem}
  \label{thm:intro-trapped}
  The boundary of the outer trapped region is a smooth outermost MOTS.
  Furthermore, it is the unique outermost MOTS.
\end{theorem}

The boundary of the outer trapped region is defined and examined in
section~\ref{sec:trapped-region}, where
theorem~\ref{thm:trapped-region} is proved, a more precise version of
theorem~\ref{thm:intro-trapped}. The main idea here is that barrier
constructions using a smoothing result from
Kriele-Hayward~\cite{Kriele-Hayward:1997}, cf.\
lemma~\ref{thm:smooth_corner}, and theorem~\ref{thm:schoen-weak} can
be used to prove a replacement for the maximum principle for outermost
MOTS. Together with the compactness properties for stable MOTS, and
the area estimate for outermost MOTS, this gives the result.

Although the presentation here is restricted to the $n=3$ dimensional
case, most of the techniques proposed generalize to higher dimensions.
The points which need to be addressed in the higher dimensional case
are regularity issues for Jang's equation, cf.\
remark~\ref{rem:regularity}, and the a priori curvature estimates for
stable MOTS used in the surgery procedure of
section~\ref{sec:surgery}. See \cite{Eichmair:2007} for a treatment of
these issues in the higher dimensional case.

%%% Local Variables: 
%%% mode: latex
%%% TeX-master: "master"
%%% End: 

%% file: prelim.tex
\section{Preliminaries}
\label{sec:preliminaries}
An initial data set for the Einstein equations is a 3-dimensional
Riemannian manifold $(M,g)$ together with a symmetric two-tensor 
$K$ representing the second fundamental form of $M$ viewed as a Cauchy 
hypersurface in a four dimensional spacetime. In this paper we will not make
further use of the spacetime geometry and in particular, energy conditions or
constraint equations on $(g,K)$ are not needed for this paper. 
%In some of the
%considerations in the paper, $M$ is allowed to have a non-empty boundary
%$\del M$ 
%consisting of compact surfaces. We assume all fields to be smooth (up to
%boundary when that is relevant), unless otherwise stated. 

A surface in $M$ is called \emph{two-sided} if its normal bundle is
orientable, i.e. if it is possible to choose a globally defined
normal.  As there are two such choices we will assume that there is
one distinguished direction which we call the \emph{outer normal}. We
will denote this outer normal vector field by $\nu$.

Given a two-sided surface $\Sigma$ in $M$, we denote its second
fundamental form, defined with respect to it outer normal $\nu$, by
$A$.  Further, we denote by $H, P$ the mean curvature, $H = \divSig \,
\nu$, and the trace of $K^\Sigma = K|_{T\Sigma}$ along $\Sigma$, $P =
\trSig K^\Sigma$, respectively. The outward null expansion of $\Sigma$
is the quantity $\theta^+ = P + H$ and the inward null expansion is
$\theta^- = P - H$. The null expansions $\theta^\pm$ are the traces of
the null second fundamental forms $\chi^\pm = K^\Sigma \pm A$.
\begin{definition}
  A smooth, embedded, compact, two-sided surface $\Sigma$ is
  a \todef{marginal\-ly outer trapped surface (MOTS)} if
  $\theta^+=0$ on $\Sigma$.
\end{definition}
Unless otherwise stated, we shall consider data sets $(M,g,K)$ with the
following properties. We assume $M$ is a compact manifold with
boundary $\del M$ such that $\del M = \del^-M \cup \del^+M$ is the
disjoint union of a possibly empty 
\emph{inner} boundary $\del^-M$, which we endow
with the normal vector field pointing into $M$ and the non-empty
\emph{outer} boundary $\del^+M$ which we endow with the normal vector
field pointing out of $M$. We assume the outer boundary is a barrier, i.e. 
$\theta^+[\del^+ M] > 0$. All fields are assumed to be smooth up to
boundary. 
\begin{definition}
  \label{def:bounding}
  We say that $\Sigma$ bounds a region $\Omega\subset M$ with respect
  to $\del^+M$, if the boundary $\del\Omega$ is the disjoint
  union $\del \Omega = \Sigma \cup \del^+ M$. 

  In this case, the normal pointing into $\Omega$ will be used as the
  outer normal for $\Sigma$.
\end{definition}
Note that if $\Sigma$ bounds with respect to $\del^+M$, then $\Sigma$
is homologous to $\del^+M$.

For the existence results, theorems~\ref{thm:schoen}
and~\ref{thm:schoen-weak}, we need a non-empty $\del^-M$ with
$\theta^+[\del^- M] < 0$ as inner barrier surface. On the other
hand, for the area bound, theorem~\ref{thm:area-bound}, and
theorem~\ref{thm:trapped-region}, which shows regularity of the
trapped region, we allow $\del^-M$ to be empty, and assume that
$\del^- M$ is a weak barrier, $\theta^+[\del^- M] \leq 0$, if
nonempty.
\begin{definition}
  \label{def:outermost}
  If $(M,g,K)$ is as before, with $\del^-M$ possibly empty, then an
  \todef{outermost MOTS} is a MOTS $\Sigma$ which bounds a region
  $\Omega$ with respect to $\del^+ M$ as in
  definition~\ref{def:bounding} with the following properties.  If
  $\Sigma'$ is a MOTS bounding a set $\Omega'$ with respect to
  $\del^+M$ with $\Omega'\subset\Omega$, then $\Omega'=\Omega$.
\end{definition}
We recall the strong maximum principle for MOTS. Note that it is only
valid if two surfaces touch with the normals pointing in the same
direction, as the surfaces have to be oriented the same way to use the
maximum principle for quasilinear elliptic equations of second order
\cite{Ashtekar-Galloway:2005,Gilbarg-Trudinger:1998}.
\begin{proposition}
  \label{prop:max_princ}
  Let $(M,g,K)$ be an initial data set and let $\Sigma_i\subset M$,
  $i=1,2$ be two connected $C^2$-surfaces touching at one point $p$,
  such that the outer normals of $\Sigma_i$ agree at $p$.  Assume
  furthermore that $\Sigma_2$ lies to the outside of $\Sigma_1$, that
  is in direction of its outer normal near $p$, and that
  \begin{equation*}
    \sup_{\Sigma_1} \theta^+[\Sigma_1] \leq \inf_{\Sigma_2} \theta^+[\Sigma_2].
  \end{equation*}
  Then $\Sigma_1=\Sigma_2$.
\end{proposition}
If $\theta^+[\del^-M]<0$ and $\theta^+[\del^+ M]>0$ then by continuity
the parallel surfaces to $\del^\pm M$, i.e. 
the level sets of the distance $\dist(\cdot,\del^{\pm} M)$, 
will satisfy the same inequality if the distance is sufficiently small. 
For later use we formalize this in the
following definition. 
\begin{definition}
  \label{def:boundary_as_barrier}
  Assume $\theta^+[\del^-M] < 0$ and $\theta^+[\del^+ M] > 0$. Denote
  by $\Sigma^\pm_s$ the parallel surface to $\del^\pm M$ at distance
  $s$. Let
  \begin{equation*}
    \rho^+(M,g,K;\del^+ M)
    :=
    \sup \big\{ s : \Sigma^+_s\ \text{is smooth, embedded and}\
    \theta^+[\Sigma^+_s] >0 \big\}
  \end{equation*}
  and
  \begin{equation*}
    \rho^-(M,g,K;\del^- M)
    :=
    \sup \big\{ s : \Sigma^-_s\ \text{is smooth, embedded and}\
    \theta^+[\Sigma^-_s] < 0 \big\}
  \end{equation*}
  where we set $\rho^-(M,g,K;\del^- M) = \infty$ if $\del^-M
  =\emptyset$. Let
  \begin{equation*}
    \rho (M,g,K; \del M)
    :=
    \min \big\{ \rho(M,g,K;\del^+ M), \rho^-(M,g,K;\del^- M) \big\}.
  \end{equation*}
\end{definition}
Note that $\rho (M,g,K; \del M)$ only depends on the geometry of
$(M,g,K)$. In fact we have 
\begin{lemma}
  \label{lem:barrdist} 
  Assume $\theta^+[\del^-M] < 0$ and $\theta^+[\del^+ M] > 0$. 
  Let $\|A\|_{C^0(\del M)}$ be the norm of
  the second fundamental form of the boundary. 
  There is a constant $C$ depending only on $\inf_{\del M} |\theta^+[\del M]|$,
  $\|K\|_{C^1(M)}$, $\|\RiemM\|_{C^0(M)}$,
  and $\|A\|_{C^0(\del M)}\,$ such that 
  \begin{equation*}
    \rho (M,g,K; \del M)^{-1} \leq C.
  \end{equation*}
\end{lemma} 
The significance of definition \ref{def:boundary_as_barrier} lies in
the following lemma, which is an immediate consequence of the strong
maximum principle.
\begin{lemma}
  \label{thm:boundary_as_barrier}
  If $(M,g,K)$ is as before, with $\del^-M$ possibly empty, and
  $\Sigma\subset M$ is a smooth MOTS homologous to $\del^+M$, then
  \begin{equation*}
    \dist(\Sigma,\del M) \geq \rho (M,g,K; \del M).
  \end{equation*}
\end{lemma}
Later, we will need the injectivity radius of $(M,g)$, restricted to
MOTS. By the previous lemma these surfaces cannot enter a collar
neighborhood of $\del M$ if $\del M$ is a barrier, and 
thus we only need to consider the
injectivity radius of points at least distance $\rho (M,g,K; \del M)$
away from $\del M$.
\begin{definition}
  \label{def:injectivity_radius}
  For $p\in M$ let $\inj (M,g; p)$ be the injectivity radius of
  $(M,g)$ at $p$. Then denote
  \begin{equation*}
    \inj_\rho(M,g,K;\del M)
    :=
    \inf\big\{ \inj (M,g; p) : \dist(p,\del M) \geq \rho (M,g,K; \del
    M) \big\}.
  \end{equation*}
\end{definition}
Let $\Sigma$ be a MOTS and let $F: \Sigma\times(\eps,\eps)\to M$ be a
normal variation of $\Sigma$, that is $F(\cdot, 0) = \id_\Sigma$ and
$\ddeval{F}{s}{s=0} = f\nu$ for a function $f\in C^\infty(\Sigma)$.
Then the variation of $\theta^+$ at $\Sigma$ is given by the operator
\begin{equation*}
  \begin{split}
    &\ddeval{\theta^+[F(\Sigma,s)]}{s}{s=0}
    =
    L_M f\\
    &\quad= -\lapSig f + 2 S(\nabSig f) + f\big(\divSig S
    -\tfrac{1}{2}|\chi^+|^2 - |S|^2 + \tfrac{1}{2}\ScalSig - \mu + J(\nu) \big).
  \end{split}
\end{equation*}
Here $\lapSig$, $\nabSig$ and $\divSig$ are the Laplace-Beltrami
operator, the tangential gradient and the divergence along $\Sigma$.
Furthermore $S(\cdot) = K(\nu,\cdot)^T$, where $(\cdot)^T$ denotes
orthogonal projection to $T\Sigma$. $\ScalSig$ is the scalar curvature
of $\Sigma$, $\mu = \tfrac{1}{2}\big( \ScalM - |K|^2 + (\tr K)^2
\big)$, and $J = \div K -\nabla\tr K$. This operator is not
self-adjoint. However, the general theory for elliptic operators of second
order implies that $L_M$ has a unique eigenvalue $\lambda$ with
minimal real part. This eigenvalue is real, and the corresponding
eigenfunction does not change sign. It is called the \emph{principal
  eigenvalue} of $L_M$. In
\cite{Andersson-Mars-Simon:2005,andersson-mars-simon:2007} the
following notion was introduced:
\begin{definition}
  \label{def:stable}
  A MOTS is called \todef{stable} if the principal eigenvalue of
  $L_M$ is non-negative.
\end{definition}
A strictly stable MOTS, that is with $\lambda>0$, can be deformed in
the direction of the outer normal such that $\theta^+>0$ on the
deformed surfaces. To see this simply use the principal eigenfunction
with the positive sign as the lapse of a normal
deformation. Analogously, unstable surfaces can be deformed in the
direction of the outer normal such that $\theta^+<0$ on the deformed
surface.

For a further discussion on stability see
\cite{Andersson-Mars-Simon:2005,andersson-mars-simon:2007,Andersson-Metzger:2005}.
We shall need theorem 1.2 from \cite{Andersson-Metzger:2005}.
\begin{theorem}
  \label{thm:curv-est} 
  Suppose $\Sigma$ is a stable MOTS in $(M,g,K)$ homologous to
  $\del^+M$.  Then the second fundamental form $A$ satisfies the
  inequality
  \begin{equation*}
    \|A\|_\infty
    \leq
    C\big(
    \|K\|_{C^1(M)},
    \|\RiemM\|_{C^0(M)},
    \inj_\rho(M,g,K; \del M)^{-1}\big)\,.    
  \end{equation*}
\end{theorem}
Note that in the reference \cite{Andersson-Metzger:2005} this theorem
is proven for $M$ without boundary. The same method gives the estimate
where the dependency $\inj(M,g)$ in the original statement is replaced
by $\inj_\rho(M,g,K;\del M)$, as this is the quantity which needs to
be controlled to apply the Hoffman-Spruck Sobolev inequality.

Subsequently we denote by $B^M_r(O)$ the open ball in $M$ with radius
$r$ around $O$, and by $B^\Sigma_r(p)$ the intrinsic open ball in
$\Sigma$.

Let $M$ be as above and let $\Sigma\subset M$, be a compact smooth
embedded two-sided surface, and let $G_\Sigma$ be the normal exponential
map of $\Sigma$:
\begin{equation}
  \label{eq:24}
  G_\Sigma
  :
  \Sigma \times \big(-\dist(\Sigma,\del M), \dist(\Sigma,\del M)\big) \to M
  :
  (p,r) \mapsto \exp^M_p(r\nu)
\end{equation}
where $\exp^M_p : T_p M \to M$ is the exponential map of $M$ at $p$.
Locally $G_\Sigma$ is injective and well behaved, this is the content
of the following well-known lemma. We shall focus on the local outer
injectivity in the following sense. We denote by $\inj(M,g;\Sigma)$
the injectivity radius on $(M,g)$ restricted to $\Sigma$.
\begin{lemma}
  \label{lemma:local_inj}
  If $\Sigma\subset M$ is as above with bounded curvature, there
  exists $0<i^+_0(\Sigma) < \inj(M,g;\Sigma)$, depending only on
  $\inj(M,g;\Sigma)$, $\|\RiemM\|_{C^0}$, and $\sup_\Sigma |A|$, such
  that for all $x\in\Sigma$ the map
  \begin{equation*}
    G_\Sigma|_{B^\Sigma_{i_0^+(\Sigma)}(x)\times [0,i_0^+(\Sigma))} :
    B^\Sigma_{i_0^+(\Sigma)}(x)\times [0,i_0^+(\Sigma)) \to M
  \end{equation*}
  is a diffeomorphism on its image, and such that the sheets
  \begin{equation*}
    \Sigma^s_{x,i_0^+(\Sigma)}:= G_\Sigma\big(B^\Sigma_{i_0^+(\Sigma)}(x), s\big)
  \end{equation*}
  are discs with bounded curvature $\sup_{\Sigma^s}|A| \leq 2\sup_\Sigma|A|$, 
  for $s\in [0,i_0^+(\Sigma))$.
\end{lemma}
This lemma reflects the local well-behavedness of the distance
surfaces to $\Sigma$, in particular including the curvature bound. In
contrast the next definition aims at the global behavior. Again, we
only focus on the outward injectivity.
\begin{definition}
  \label{def:global_inj}
  The \todef{outer injectivity radius} of $\Sigma$ is
  \begin{equation*}
    i^+(\Sigma):= \sup \big\{ \delta : G_\Sigma|_{\Sigma\times [0,\delta)} \to M
    \text{\ is injective\ } \big\}.
  \end{equation*}
\end{definition}
It is intuitively clear that if  
$i^+(\Sigma)$ is smaller than $i_0^+(\Sigma)$, then the surface
nearly meets itself on the outside. A precise formulation is given by the
following lemma. 
\begin{figure}[Ht!]
  \centering
  \resizebox{.6\linewidth}{!}{\input{pics/touch_pic.tex}}
  \caption{A surface that nearly meets itself.}
\label{fig:touch}
\end{figure}
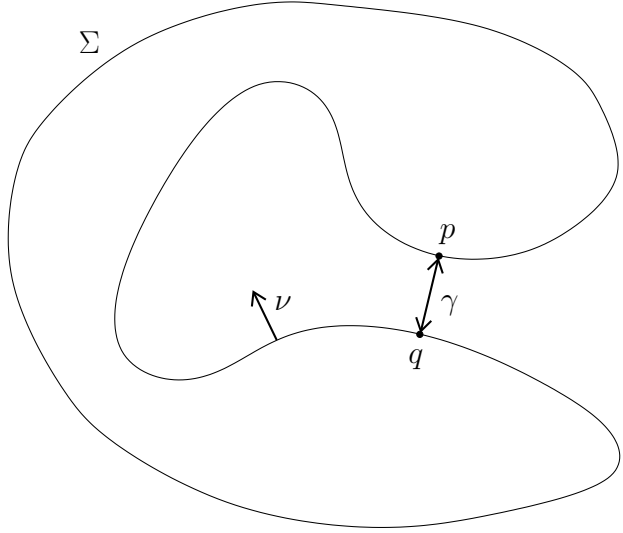
\begin{lemma}
  \label{lemma:almost_touch}
  Let $\Sigma$ be a compact, embedded and two-sided surface with
  $i^+(\Sigma) < \tfrac{1}{2}i_0^+(\Sigma)$. Then there exist two
  points $p,q\in\Sigma$ with $\distM(p,q) = 2i^+(\Sigma)$ but
  $\distSig(p,q) \geq i_0^+(\Sigma) > 2i^+(\Sigma)$.
    
  The points $p$ and $q$ can be joined by a geodesic segment $\gamma$
  in $M$, which is orthogonal to $\Sigma$ at $p$ and $q$ and as a set
  \begin{equation*}
    \gamma
    =
    G_\Sigma|_{B^\Sigma_{i_0^+(\Sigma)}(p)\times[0,i_0^+(\Sigma))}
    (p, [0,2 i^+])
    = G_\Sigma|_{B^\Sigma_{i_0^+(\Sigma)}(q)\times[0,i_0^+(\Sigma))}
    (q, [0,2 i^+]).  
  \end{equation*}  
\end{lemma}
\begin{proof}
  From the definition of $i^+$ we know that 
  \begin{equation*}
    G_\Sigma (\cdot,  i^+(\Sigma) ) : \Sigma \to M    
  \end{equation*}
  is not injective. Thus there exist two points $p,q\in\Sigma$ which
  map to the same point $O\in M$. By lemma~\ref{lemma:local_inj}
  $\distSig(p,q)\geq i_0^+(\Sigma)$. Furthermore $O$ has distance
  $i^+(\Sigma)$ to $\Sigma$ and to $p,q$ so $\dist(O,\Sigma)
  =\dist(O,p)$ and hence the geodesic segment $\gamma_p$ joining $O$
  to $p$ is perpendicular to $\Sigma$. Similarly the geodesic segment
  $\gamma_q$ joining $O$ and $q$ is perpendicular to $\Sigma$. Thus
  $\dist (p,q) \leq 2i^+(\Sigma)$. If $\dist(p,q) < 2i^+(\Sigma)$ then
  there would be a parallel surface to $\Sigma$ at distance
  $d<i^+(\Sigma)$ which intersects itself, which is not possible as
  $G_\Sigma(\cdot,d)$ is injective. Thus $\dist(p,q) = 2d$ and
  $\gamma_p$ and $\gamma_q$ must form a smooth geodesic, as otherwise
  the angle at $O$ could be smoothed out to yield a shorter geodesic.
\end{proof}
Figure \ref{fig:touch} shows the situation in the lemma. 
It follows from the definition of $i^+(\Sigma)$ that 
the points $p,q$ minimize the distance between the sheets
$B^\Sigma_{i_0^+(\Sigma)}(p)$ and $B^\Sigma_{i_0^+(\Sigma)}(q)$, and hence 
$\gamma$ is orthogonal 
to $\Sigma$ at $p$ and $q$.  In addition $\gamma$ does not
intersect $\Sigma$ in any other points except $p$ and $q$. If we parameterize
$\gamma$ by arc 
length as a curve joining $p$ to $q$, the tangent to
$\gamma$ at $p$ coincides with the normal $\nu$ to $\Sigma$. Similarly, with 
$\gamma$ arc 
length parameterized as a curve joining $q$ to $p$, the tangent to
$\gamma$ at $q$ coincides with the normal $\nu$ to $\Sigma$ at $q$.
This means that $\gamma$ lies completely on the outside of $\Sigma$.

For later reference, we need the following smoothing result from
\cite[Lemma 6]{Kriele-Hayward:1997}.
\begin{lemma}
  \label{thm:smooth_corner}
  Let $\Sigma_1,\Sigma_2\subset M$ be smooth two-sided surfaces which
  intersect transversely in a smooth curve $\gamma$. Choose one
  connected component $\Sigma^\pm$ of each set $\Sigma_i\setminus
  \gamma$ such that the outer normals $\nu^\pm$ of these components
  satisfy $g(\nu^+,\nu^-) \leq 0$ along $\gamma$. Then for any
  neighborhood $U$ of $\gamma$ there exists a smooth surface $\Sigma$
  and a continuous and piecewise smooth bijection $\Phi: \Sigma^+ \cup
  \Sigma^-\cup \gamma\to\Sigma$ such that
  \begin{enumerate}
  \item $\Phi(x) = x$ for all $x\in(\Sigma^+\cup\Sigma^-)\setminus U$,
  \item ($\Sigma^+\cup\Sigma^-)\setminus U = \Sigma \setminus U$, and
  \item $\theta^+[\Sigma](x) \leq \theta^+[\Sigma^+](x)$ for $x\in
    \Sigma^+$ and $\theta^+[\Sigma](x) \leq \theta^+[\Sigma^-](x)$ for
    $x\in \Sigma^-$.
  \end{enumerate}
  Moreover $\Sigma$ lies in the connected component of $U\setminus
  (\Sigma^+\cup\Sigma^-\cup\gamma)$ into which $\nu^\pm$ point.
\end{lemma}
Briefly stated, this procedure works by replacing the inward corner
near $\gamma$ by a smooth patch with $\theta^+$ very negative. The
reason why this procedure works is that the corner is a concentration
of negative mean curvature, that is negative $\theta^+$.

%%% Local Variables: 
%%% mode: latex
%%% TeX-master: "master"
%%% End: 

%% file: pics/touch_pic.tex
\begin{picture}(0,0)%
\includegraphics{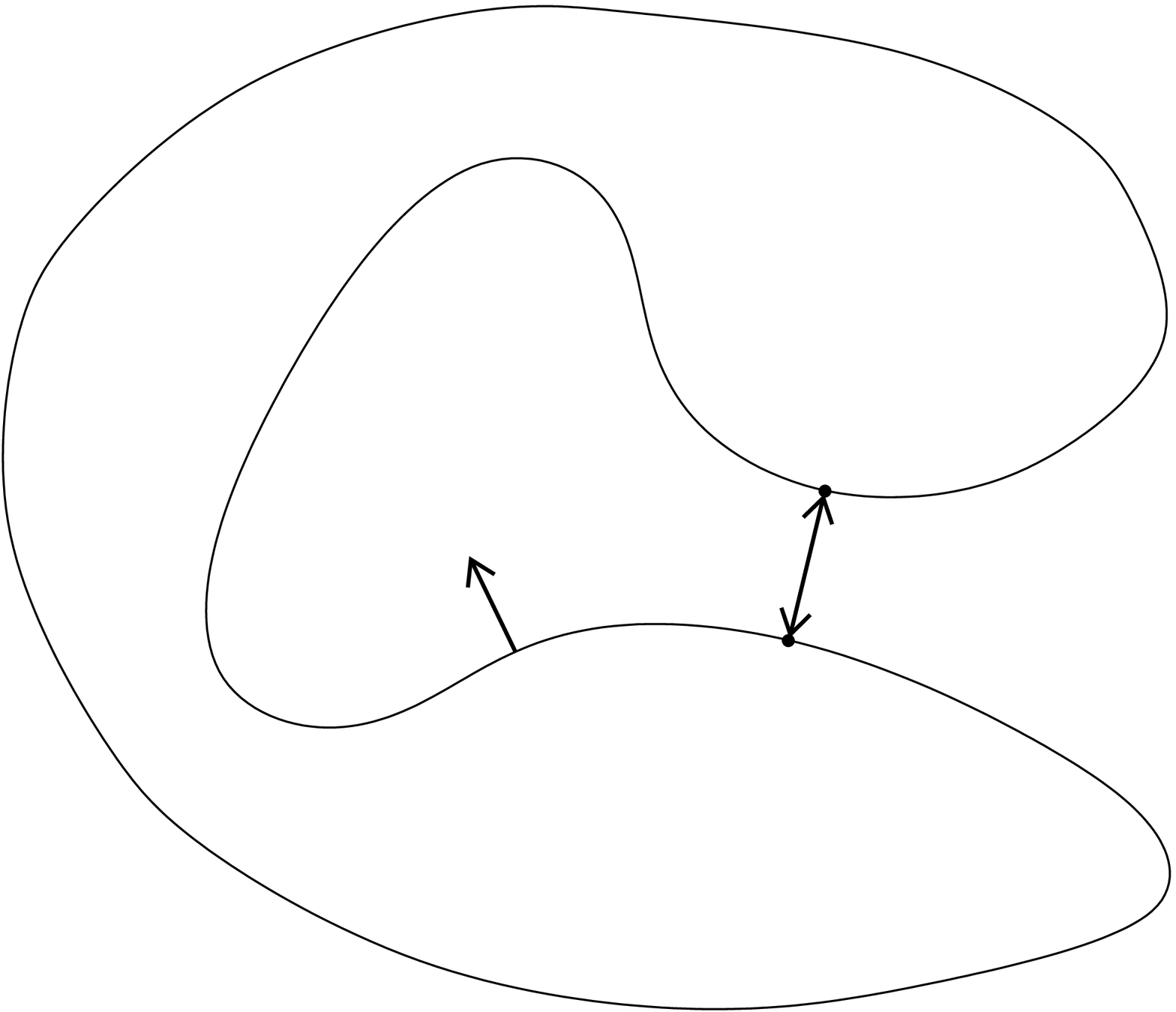}%
\end{picture}%
\setlength{\unitlength}{4144sp}%
\begingroup\makeatletter\ifx\SetFigFont\undefined%
\gdef\SetFigFont#1#2#3#4#5{%
  \reset@font\fontsize{#1}{#2pt}%
  \fontfamily{#3}\fontseries{#4}\fontshape{#5}%
  \selectfont}%
\fi\endgroup%
\begin{picture}(8896,7657)(1371,-6934)
\put(2417,-40){\makebox(0,0)[lb]{\smash{{\SetFigFont{29}{34.8}{\familydefault}{\mddefault}{\updefault}{\color[rgb]{0,0,0}$\Sigma$}%
}}}}
\put(5247,-3762){\makebox(0,0)[lb]{\smash{{\SetFigFont{29}{34.8}{\familydefault}{\mddefault}{\updefault}{\color[rgb]{0,0,0}$\nu$}%
}}}}
\put(7653,-3750){\makebox(0,0)[lb]{\smash{{\SetFigFont{29}{34.8}{\familydefault}{\mddefault}{\updefault}{\color[rgb]{0,0,0}$\gamma$}%
}}}}
\put(7638,-2713){\makebox(0,0)[lb]{\smash{{\SetFigFont{29}{34.8}{\familydefault}{\mddefault}{\updefault}{\color[rgb]{0,0,0}$p$}%
}}}}
\put(7186,-4525){\makebox(0,0)[lb]{\smash{{\SetFigFont{29}{34.8}{\familydefault}{\mddefault}{\updefault}{\color[rgb]{0,0,0}$q$}%
}}}}
\end{picture}%

%%% Local Variables: 
%%% mode: latex
%%% TeX-master: "master"
%%% End: 

%% file: existence.tex
\section{Existence of MOTS}
\label{sec:theorem-schoen}
This section is devoted to a proof of Schoen's existence theorem for
MOTS \cite{Schoen:2004} in the presence of barrier surfaces.
\begin{theorem}
  \label{thm:schoen}
  Let $(M,g,K)$ be a smooth, compact initial data set with $\del M$ the
  disjoint union $\del M = \del^-M \cup \del^+ M$ such that $\del^\pm
  M$ are non-empty, smooth, compact surfaces without boundary and
  $\theta^+[\del^-M]< 0$ with respect to the normal pointing into
  $M$ and $\theta^+[\del^+M]>0$ with respect to the normal pointing
  out of $M$.
  Then there exists an non-empty, smooth, embedded MOTS $\Sigma$
  homologous to $\del^+ M$.
\end{theorem}
\begin{remark}
  \label{rem:regularity}
  The proof presented here readily carries over to $n$ dimensional $M$
  with $3\leq n\leq 5$. The dimensional restriction is due to the
  method used for the curvature estimates in
  proposition~\ref{p:curv_est} in \cite{Schoen-Yau:1981}. Higher
  dimensional replacements for this proposition are accessible via
  methods from geometric measure theory, cf.\ \cite{Eichmair:2007}.
\end{remark}
\subsection{Setup and Outline}
Consider $\bar M := M\times\IR$ equipped with the metric $\bar g = g
+ dz^2$, and define $\bar K$ on $\bar M$ as the pull-back of $K$
under the projection $\pi: M\times\IR \to M :(p,z)\mapsto p$.  For a
function $f$ on $M$ we consider $N = \graph f := \{ (p, f(p)) : p\in
M \}$, with induced metric $\bar g$, which is of the form
\begin{equation*}
  \bar g_{ij} = g_{ij} + \nabla_i f \nabla_j f,\qquad
  \bar g^{ij} = g^{ij} - \frac{\nabla^i f \nabla^j f}{1 + |\nabla f |^2}.
\end{equation*}
The mean curvature of $N$ with respect to the downward normal is
\begin{equation*}
  \CH[f]
  =
  \div \left( \frac{\nabla f}{\sqrt{1+|\nabla f|^2}} \right).
\end{equation*}
Furthermore let
\begin{equation*}
  \CP[f] = \tr_N \bar K
\end{equation*}
be the trace of $\bar K$ taken along $N$. Now we can write Jang's
equation as
\begin{equation}
  \label{eq:Jang}
  \CJ[f] = \CH[f] - \CP[f] = 0.
\end{equation}
We shall consider the Dirichlet problem for this equation with
boundary values $f \big{|}_{\del^\pm M} = \mp Z$, for constants
$Z > 0$.

Equation \eqref{eq:Jang} is a quasilinear elliptic equation of
divergence form. In particular, it is a prescribed mean curvature
equation with gradient dependent lower order term. For such equations
the strong maximum principle does not apply directly to give upper and lower
bounds for the solution, without assuming extra conditions for example
on the size of the domain. Further, the boundary gradient estimates
needed for the proof of existence of classical solutions typically
require restrictions on the geometry of the boundary. Therefore we
cannot prove existence of solutions to the Dirichlet problem directly
for equation \eqref{eq:Jang}. In general it is to be expected that
solutions to the Dirichlet problem blow up in the interior.

We follow the approach of \cite{Schoen-Yau:1981} and 
regularize Jang's equation 
by adding a capillarity term. Thus we consider instead of  \eqref{eq:Jang},
the equation
\begin{equation} \label{eq:Jang-mod}
\CJ_\tau[f] = \CJ[f] - \tau f =0
\end{equation}
for $\tau > 0$. After suitably modifying the data, we are able to
apply Leray-Schauder theory \cite{Gilbarg-Trudinger:1998} 
to prove existence of solutions to the
Dirichlet problem.  Letting $\tau \to 0$ gives a sequence of solutions
which by uniform curvature estimates for $\graph f_\tau$ has a
subsequence which converges to a solution of Jang's equation (which in
general may have blowups).

The goal is in fact to prove existence of MOTS by constructing a
blowup solution to Jang's equation. For this purpose, we set $Z =
\delta/\tau$ for a suitable $\delta$ and let $\tau \to 0$.

A key observation of \cite{Schoen-Yau:1981} is that solutions to
\eqref{eq:Jang-mod} 
satisfy interior estimates for the second fundamental
form, uniformly in $\tau$. These estimates allow us to pick out a
subsequence of solutions which converges to a blowup solution of Jang's
equation. After
applying a sequence of renormalizations using the fact that Jang's
equation is translation invariant, we get a vertical solution, which
projects to a MOTS on $M$.

The last part of the argument proceeds exactly as in
\cite{Schoen-Yau:1981}, and therefore the only thing which needs to be
discussed here is the Dirichlet problem.

\subsection{Preparing the Data}
\label{sec:preparing-data}
We will assume that $(M,g,K)$ is embedded into a four-dimensional
Lorentz manifold $(L,h)$ such that $g$ and $K$ are the first and
second fundamental forms of $M$ induced by $h$. As we do not require
the dominant energy condition to hold, it is rather simple to produce
an extension $(L,h)$ of $(M,g,K)$. To this end extend $g$ to
$M\times\IR$ by setting $g_t = g + tK$ on the slice $M\times{t}$. As
$K$ is symmetric, so is $g_t$ and there exists $t_0>0$ such that $g_t$
is positive definite for $t\in (-t_0,t_0)$. Then define $h$ on
$L:=M\times(-t_0,t_0)$ to be
\begin{equation*}
  h = -dt^2 + g_t.
\end{equation*}
This is a Lorentz metric and obviously induces $g$ as first
fundamental form on the slice $M_0 = M\times\{0\}$. That $K$ is the
second fundamental form follows from the second variation formula,
which implies that the second fundamental form of $M_0$ is given by
\begin{equation*}
  \ddeval{}{t}{t=0} g_t = K.
\end{equation*}
Let $t$ be a time function on $L$ with $M=\{t=0\}$ and
$s^+(x):=\dist(x,\del^+M)$ the distance function to $\del^+M$. For
small $s,t$, let $\Sigma^+_{s,t}$ be the surface given by the
intersection of the level sets of $s^+$ and $t$.  Let $n$ be the
timelike normal of the $t$-level sets and let $\nu$ be the spacelike
normal of the $s^+$-level sets, inside the $t$-levels, extending the
outward pointing normal on $\del^+M$.  This defines normal fields
$n,\nu$ at the surfaces $\Sigma^+_{s,t}$ as well as the corresponding
null normals $l^{\pm} = n \pm \nu$. For small $s,t$, we have
$\theta^+[\Sigma^+_{s,t}] > 0$.

Now perform a Lorentz rotation of the normals $n,\nu$ to get
\begin{equation*}
  \tilde \nu = \cosh \alpha \nu + \sinh \alpha n,
  \quad
  \tilde n = \sinh \alpha \nu + \cosh \alpha n.
\end{equation*}
Let $\two_{ab}^\mu$ be the second fundamental form of the surfaces
$\Sigma^+_{s,t}$ so that $H = h^{ab} \la \two_{ab}, \nu\ra$ and $P =
h^{ab} \la\two_{ab}, n \ra$, where $h_{ab}$ is the metric on
$\Sigma^+_{s,t}$.  Then with respect to the normals $\tilde \nu,
\tilde n$ we have
\begin{equation*}
  \tilde H = \cosh \alpha H + \sinh \alpha P,
  \quad
  \tilde P = \sinh \alpha H + \cosh \alpha P
\end{equation*}
and the corresponding null expansions 
\begin{equation*}
\tilde \theta^{\pm} = \tilde P \pm \tilde H
\end{equation*}
are given by 
\begin{equation*}
\tilde \theta^{\pm} =  e^{\pm \alpha} \theta^{\pm}.
\end{equation*}
Further we note 
\begin{align*} 
\tilde H &= \half e^\alpha \theta^+ - \half e^{-\alpha} \theta^-, \\
\tilde P &= \half e^\alpha \theta^+ + \half e^{-\alpha} \theta^-.
\end{align*} 
Deform $M$ to $\tilde M$ by bending up along the outgoing future light
cone at $\del^+M$. By doing so, we get the spacelike and timelike
normals to agree with $\tilde \nu, \tilde n$ for any $\alpha$. As the
deformed $\tilde M$ approaches the light cone, we have $\alpha \to
\infty$.  Therefore there is an $\alpha$ such that $\tilde H, \tilde
P$ are arbitrarily close to $\half e^{\alpha} \theta^+$. In
particular, if $\theta^+ > 0$, we can achieve that both $\tilde H$ and
$\tilde P$ are positive near the outer boundary of $\tilde M$.

We can proceed similarly at the inner boundary $\del^-M$, where
$\theta^+<0$ with respect to the \emph{inward pointing} normal. This means
that $\theta^-<0$ with respect to the \emph{outward pointing} normal.  Then we
can proceed as above, bending along the {\em past inward} lightcone.
This will result in $\tilde H > 0, \tilde P < 0$ (where now $\tilde H$
is defined with respect to the outward normal of $M$ as usual).

This constructs a deformed Cauchy data set $(\tilde M, \tilde g,
\tilde K)$. Let $\del \tilde M$ be the boundary of $\tilde M$
constructed by bending as above. Clearly the boundary $\del\tilde M$
is the union $\del\tilde M =\del^-\tilde M \cup\del^+\tilde M$, with
$\tilde H > 0$ on $\del\tilde M$ and $\tilde P>0$ on $\del^+\tilde M$,
$\tilde P<0$ on $\del^-\tilde M$. Let
\begin{equation*}
  \Sigma^\pm_s
  :=
  \big\{ x\in\tilde M : \dist(x,\del^\pm\tilde M) = s \big\}
\end{equation*}
the parallel surfaces to $\del^\pm\tilde M$ and
\begin{equation*}
  U^\pm_s
  :=
  \big\{ x\in\tilde M : \dist(x,\del^\pm\tilde M) < s \big\}
\end{equation*}
be the respective tubular neighborhoods. Given $\eps>0$, there exists
$\delta>0$ such that we can ensure the following properties:
\begin{equation}
  \label{eq:10}
  \begin{split}    
    &
    \begin{aligned}
      \theta^+[\Sigma^-_s] &< 0
      \\
      H[\Sigma^-_s] &> \delta
      \\
      P[\Sigma^-_s] &\leq 0
    \end{aligned}
    \quad
    \begin{aligned}
      \text{and}
      \\
      \text{and}
      \\
      \text{and}
    \end{aligned}
    \quad
    \begin{aligned}
      \theta^+[\Sigma^+_s] &> 0
      \\
      H[\Sigma^+_s] &> \delta
      \\
      P[\Sigma^+_s] &\geq 0
    \end{aligned}
    \quad
    \begin{aligned}
      &\text{for}\quad s\in [0,4\eps],
      \\
      &\text{for}\quad s\in[0,2\eps],
      \\
      &\text{for}\quad s\in[0,2\eps],
    \end{aligned}
    \\
    &
    \text{the data is unchanged in}\quad M_{3\eps}.
  \end{split}
\end{equation}
We abuse notation here by computing $H$ with respect to the outward
pointing normal for $\del\tilde M$, but compute $\theta^+$ still with
respect to the inward pointing normal near $\del^-\tilde M$, which makes
$\theta^+ = P - H$ near $\del^-\tilde M$.

Fix such an $\eps > 0$ and let $\zeta(s)$ be a non-negative cutoff
function on $s \geq 0$, such that $\zeta(s) = 0$ for $s \in [0,\eps]$,
$\zeta(s) > 0$ for $s > \eps$, and $\zeta(s) = 1$ for $s \geq 2\eps$.
Now define $\zeta(x) = \zeta(d(x,\partial \tilde M))$, and consider
the data set $(\tilde g, \zeta \tilde K)$. From now on we
denote this data set by $(M, g, K)$. The important point to
note here is that this final cut-off does not affect the first
property of \eqref{eq:10}, so that we still retain the barrier effect
of the boundary. We find that with respect to the cut-off data we have
the following properties near the boundary:
\begin{equation}
  \label{eq:14}
  \begin{split}    
    &
    \begin{aligned}
      \theta^+[\Sigma^-_s] &< 0
      \\
      H[\Sigma^-_s] &> \delta
    \end{aligned}
    \quad
    \begin{aligned}
      \text{and}
      \\
      \text{and}
    \end{aligned}
    \quad
    \begin{aligned}
      \theta^+[\Sigma^+_s] &> 0
      \\
      H[\Sigma^+_s] &> \delta
    \end{aligned}
    \quad
    \begin{aligned}
      &\text{for}\quad s\in [0,4\eps],
      \\
      &\text{for}\quad s\in[0,2\eps],
    \end{aligned}
    \\
    &
    K\equiv 0
    \quad\text{in}\quad U_\eps,
    \quad\text{and}
    \\
    &
    \text{the data is unchanged in}\quad M_{3\eps}.
  \end{split}
\end{equation}
\subsection{Existence Proof}
In order to construct solutions to the Dirichlet problem for 
\eqref{eq:Jang-mod}, we consider,
following \cite{Schoen-Yau:1981}, the family of equations
\begin{equation}
  \label{eq:stjang}
  \CH[f] - \sigma \CP[f] = \tau f, \quad f \big{|}_{\partial M} = \sigma \phi
\end{equation} 
for $\sigma \in [0,1]$ and $\tau \in [0,1]$. We need the following
estimates.
\begin{proposition}
  \label{p:curv_est}
  Let $N$ be the graph of a function $f$ satisfying the
  equation
  \begin{equation*}
    \CH[f]-\sigma \CP[f] = F \quad\text{in}\quad M
  \end{equation*}
  with $F\in C^1(\bar M)$, then the second fundamental form $A$ of $N$
  satisfies the estimate
  \begin{equation*}
    |A|(p,f(p)) \leq C\big( \|\RiemM\|_{C^0}, \|K\|_{C^1},
    \dist(p,\del M)^{-1}, \inj(M,g,p)^{-1}, \|F\|_{C^1} \big).
  \end{equation*}
  In fact, if we extend the normal $\bar\nu$ of $N$ to $M\times\IR$,
  then
  \begin{equation*}
    |\bar \nabla \bar \nu|(p,t) \leq C\big( \|\RiemM\|_{C^0}, \|K\|_{C^1},
    \dist_M(p,\del M)^{-1}, \inj(M,g,p)^{-1}, \|F\|_{C^1} \big).
  \end{equation*}
\end{proposition}
\begin{proof}
  This is analogous to \cite[Proposition 1 and Proposition 2]{Schoen-Yau:1981}.
\end{proof}
\begin{proposition}
  \label{p:sup-est}
  Let $f_{\sigma,\tau}$ be a solution to \eqref{eq:stjang} with
  parameters $\Sigma$ and $\tau$. Then $f_{\sigma,\tau}$ satisfies the estimates
  \begin{equation*}
    \sup_M |f_{\sigma,\tau}| \leq \max \big\{ 3\|K\|_{C^0}/\tau,
    \sup_{\del M}|\phi| \big\},
  \end{equation*}
  and
  \begin{equation*}
    \sup_M |\nabla f_{\sigma,\tau}| \leq \max \big\{
    c(\|\RiemM\|_{C^0}+\|\nabla K\|_{C^0})/\tau, \sup_{\del M}|\nabla
    f_{\sigma,\tau}| \big\}
  \end{equation*}
\end{proposition}
\begin{proof}
  This follows from the maximum principle, as in
  \cite[Section 4]{Schoen-Yau:1981}.
\end{proof}
Hence we can estimate the gradient once we have a boundary gradient
estimate.
\begin{proposition}
  \label{p:boundary-est}
  Let $(M,g,K)$ be a data set such that there are $\eps > 0$, $\delta
  > 0$, such that for $s \in [0,\eps]$ the surfaces
  \begin{equation*}
    \Sigma_s := \{ p\in M : \dist (p,\del M)= s \}
  \end{equation*}
  satisfy $H>\delta$. Further, assume that $K\equiv 0$ in $\{ p: \
  \dist(p,\del M) < \eps\}$. Let $f_{\tau,\sigma}$ be a solution of
  \begin{equation*}
    \CJ_{\tau,\sigma}[f_{\tau,\sigma}] = \CH[f_{\tau,\sigma}] - \sigma
    \CP[f_{\tau,\sigma}] - \tau  f_{\tau,\sigma} = 0,
  \end{equation*}
  such that $f_{\tau,\sigma}$ is constant on each component of $\del
  M$. Suppose that
  \begin{equation*}
    \sup_M |f_{\tau,\sigma}| = m < \infty \quad \text{and}\quad
    \sup_{\del M} |f_{\tau,\sigma}| \leq \tfrac{\delta}{2\tau}.    
  \end{equation*}
  Then
  \begin{equation*}
    \sup_{\del M} |\nabla f_{\tau,\sigma}| \leq  \max\{
    \tfrac{1}{\sqrt{3}}, 2\eps^{-1} m \}.
  \end{equation*}
\end{proposition}
\begin{proof}
  We proceed by constructing a barrier near $\del^- M$. Consider
  functions $w$ of the form
  \begin{equation*}
    w = \psi(s) \qquad s = \dist(\cdot, \del^- M).
  \end{equation*}
  where $\psi:[0,\eps]\to \IR$ is a scalar function. For functions of
  this form we have
  \begin{equation}
    \label{eq:9}
    \CJ_{\tau,\sigma}[w] = - \frac{\psi'}{(1 + (\psi')^2)^{1/2}}
    H[\Sigma_s] + \frac{\psi''}{(1 + (\psi')^2)^{1/2}} - \tau \psi
  \end{equation}
  in the neighborhood where $K\equiv 0$. To construct an upper barrier
  near one component $\Sigma$ of $\del^- M$, set $w^+ := \psi^+(s)$ with
  $\psi^+(s) = a + bs$, where $a$ is the value of $f_{\tau,\sigma}$ on
  $\Sigma$. We can then pick $b$ so large that
  $\tfrac{b}{(1+b^2)^{1/2}}\geq \tfrac{1}{2}$, that is
  $b\geq\tfrac{1}{\sqrt{3}}$. Then \eqref{eq:9} yields that
  \begin{equation*}
    \begin{split}
      \CJ_{\tau,\sigma}[w^+]
      &\leq
      -\tfrac{\delta}{2} + \tau |a| - \tau b s
      \\
      &\leq
      -\tfrac{\delta}{2} + \tau \sup_{\del M}|f| - \tau b s
      \leq -\tau b s \leq 0.
    \end{split}
  \end{equation*}
  We can then choose $b$ so large that $a+b\eps\geq m$, that is
  $b\geq 2\eps^{-1}m$. Thus we have constructed an upper barrier, the
  construction of the lower barrier is analogous.

  The barrier near $\del^+ M$ can be constructed analogously, using
  the expression
  \begin{equation}
    \label{eq:2}
    \CJ_{\tau,\sigma}[w] = \frac{\psi'}{(1 + (\psi')^2)^{1/2}}
    H[\Sigma_s] + \frac{\psi''}{(1 + (\psi')^2)^{1/2}} - \tau \psi
  \end{equation}
  for $\CJ_{\tau,\sigma}$ near $\del^+ M$.
\end{proof}
As a corollary, we find that given suitable boundary data, equation
\eqref{eq:stjang} is uniformly elliptic, where the ellipticity
constant does not depend on $\sigma\in[0,1]$. Thus we conclude that there
exists a solution to \eqref{eq:stjang} with $\sigma=1$ and
$\tau>0$ for such data by applying Leray-Schauder theory.
\begin{corollary}
  \label{c:existence}
  Let $(M,g,K)$ and $\phi\in C^\infty(\del M)$ be as in
  proposition~\ref{p:boundary-est}. Then the equation
  \begin{equation}
    \label{eq:jangt}
    \begin{cases}
      \CH[f_\tau] - \CP[f_\tau] = \tau f_\tau
      \\
      f |_{\del M} = \phi
    \end{cases}
  \end{equation}
  has a solution $f_\tau$ in $C^{2,\alpha}(\bar M)$ with
  \begin{equation*}
    \| f \|_{C^{2,\alpha}(\bar M)} \leq C/\tau,
  \end{equation*}
  where the constant $C=C\big(\|\RiemM\|_{C^{0,\alpha}},
  \|K\|_{C^{1,\alpha}}, \eps^{-1}\big)$.
\end{corollary}
\begin{proof}
  This is analogous to \cite[Lemma 3]{Schoen-Yau:1981}
\end{proof}
We now specify the precise data on $\del M$. Set
\begin{equation*}
  \phi =
  \begin{cases}
    \phantom{-}\tfrac{\delta}{2\tau} &\quad\text{on}\quad \del^-M \\
    -\tfrac{\delta}{2\tau}&\quad\text{on}\quad \del^+M
  \end{cases},
\end{equation*}
where $\delta$ is as in proposition~\ref{p:boundary-est}.  We then
solve \eqref{eq:jangt} with this data to obtain a family of functions
$f_\tau$. Note that the gradient estimate forces $f_\tau$ to be uniformly
large near the boundary. Denote $M_\eps = \{ p\in M : \dist(p,\del M)
> \eps\}$.
\begin{lemma}
  \label{lem:bdry_blowup}
  There exists an $\eps'>0$ such that the functions $f_\tau$ satisfy
  \begin{equation*}
    |f_\tau| \geq \tfrac{\delta}{4\tau}\qquad\text{in}\qquad
    M\setminus M_{\eps'}.
  \end{equation*}
\end{lemma}
As in~\cite[Section 4]{Schoen-Yau:1981} we can now use the curvature
estimate from proposition~\ref{p:curv_est} to obtain a limit for $\graph
f_\tau$ as $\tau\to 0$. By the previous lemma we can restrict
ourselves to $M_{\eps'}$ away from the boundary, as $f_\tau\to\infty$
uniformly on $M\setminus M_{\eps'}$. 
This gives the following result. 
\begin{proposition}
  \label{p:limit}
  There exists a sequence $\tau_i\to 0$ such that $\graph f_{\tau_i}$
  in $M_{\eps'}$ converges to a smooth manifold $N_0$ satisfying
  $H+P=0$. $N_0$ consists of a disjoint collection of components,
  which are either graphs or cylinders over compact surfaces $\Sigma$.

  Let $\Omega_\pm := \{ p : f_{\tau_i}(p)\to\pm\infty\}$ and $\Omega^0
  := \{ p : \sup_{i\geq 1}|f_{\tau_i}(p)| < \infty \}$.  Then $M$ is a
  disjoint union $M = \Omega^0 \cup \Omega^+ \cup \Omega^-$.  The set
  $\Sigma:=\del\Omega^-\setminus \del^+ M$ consists of marginally
  trapped surfaces with $\theta^+ =0$ with respect to the normal
  pointing into $\Omega^-$.
\end{proposition}
The fact that $\Sigma$ satisfies $\theta^+ =0$, can be seen as
follows. Since the $f_{\tau_i}$ converge to $-\infty$ in $\Omega^-$
and are bounded below outside of $\Omega^-$, there are just two
possibilities for the convergence of $N_{\tau_i} = \graph
f_{\tau_i}$to $N_0$ near each component $\Sigma'$ of $\Sigma$.  The
first possibility is that $\Sigma'$ is the interface between
$\Omega^+$ and $\Omega^-$. Then $N_0$ has a cylindrical component
$\Sigma'\times\IR$, and the convergence is such that the downward
normal $\bar\nu_\tau$ of $N_{\tau_i}$ converges to the normal of
$\Sigma'$ pointing out of $\Omega^-$. As $N_0$ satisfies $\CH[N_0] -
\CP[N_0] =0$ with respect to the limit of $\bar\nu_{\tau_i}$, this
implies that $H - P =0$ on $\Sigma'$ with respect to the outward
pointing normal, and hence $\theta^+ = P + H =0$ with respect to the
inward pointing normal as claimed. The second possibility is that
$\Sigma'$ is an interface between $\Omega^0$ and $\Omega^-$. Then near
$\Sigma'$, $N_0$ is a graph over $\Omega^0$ which asymptotes to
$\Sigma'\times\IR$, and since $f_{\tau_i}\to-\infty$ in $\Omega^-$,
this graph goes to $-\infty$ near $\Sigma'$ as well. Again we can
conclude that $\bar\nu_{\tau_i}$ converges to the normal of $N_0$
pointing out of $\Omega^-$.  Furthermore, $\CH - \CP = 0$ on
$\Sigma'\times\IR$ with respect to this normal, as it is the limit of
$N_0$, which satisfies $\CH -\CP =0$.  Hence we again conclude that
$\theta^+[\Sigma']=0$.

From Lemma~\ref{lem:bdry_blowup} we know that $\Omega^+$ contains a
neighborhood of $\del^-M$ and $\Omega^-$ contains a neighborhood of
$\del^+ M$, so neither one of them is trivial. In particular
$\del\Omega^-$ is the disjoint union $\del\Omega^- = \Sigma \cup
\del^+M$, where $\Sigma \subset M$ is contained in the interior of
$M$.

Recall that we had to modify the data for the
existence proof. We now show that $\Sigma$ can not enter the region
where we modified the data. To see this, note that a neighborhood of
$\del^-M$ is foliated by surfaces $\Sigma^-_s$ with
$\theta^+[\Sigma^-_s]<0$. If $\Sigma$ enters this region there
is a minimal $s$, with $\Sigma^-_s\cap\Sigma\neq\emptyset$. This
surface touches $\Sigma$ with their 
outward normals pointing in the same
direction. Thus, by the strong maximum principle, $\Sigma=\Sigma^-_s$, a
contradiction.  Furthermore, there is a neighborhood of $\del^+M$
foliated by surfaces $\Sigma^+_s$ with $\theta^+[\Sigma^+_s]>0$. We
can then proceed analogously to get a contradiction to $\Sigma$
entering this neighborhood. As data set is modified only 
in the neighborhoods discussed above, 
we find that $\Sigma$ lies entirely
in the region where the data is unchanged.

We thus conclude the proof of theorem~\ref{thm:schoen} by finding our
solution $\Sigma$ in the unmodified region of $(M,g,K)$.

It is an interesting possibility that the existence theory developed
here for the Dirichlet problem for Jang's equation can be used to
generalize Yau's result in \cite[Theorem 5.2]{Yau:2001} to more general boundary
geometries. This possibility will be investigated by the authors in
future work.

%%% Local Variables: 
%%% mode: latex
%%% TeX-master: "master"
%%% End: 

%% file: stable.tex
\section{Blowup surfaces are stable} 
\label{sec:stab-constr-mots}
While not actually necessary for the main result of the paper, we
present an extension of the results of
section~\ref{sec:theorem-schoen}.  From the arguments in
\cite{Schoen-Yau:1981} it is clear that $\Sigma$ has only components
which are symmetrized stable, where symmetrized stable refers to
non-negativity of the operator (cf. \cite{Galloway-Schoen:2006})
\begin{equation*}
  \tilde L_M f
  =
  -\lapSig f
  + f\big( \tfrac{1}{2}\ScalSig - \tfrac{1}{2}|\chi|^2  - \mu + J(\nu) \big).
\end{equation*}
Here we want to show that they are in fact stable in the sense of MOTS.
\begin{theorem}
  \label{thm:schoen-stable}
  The surface $\Sigma$ constructed in the proof of
  theorem~\ref{thm:schoen} is a \emph{stable} MOTS.
\end{theorem}
\begin{remark}
  By the same argument we can prove that any blow-up surface obtained
  by the capillarity term regularization of Jang's equation is a
  stable surface, in particular those in \cite{Schoen-Yau:1981}. Note
  that all of these surfaces are MOTS provided one chooses the right
  orientation of the normal.
\end{remark}
\begin{proof}
  The stability of $\Sigma$ will follow from a barrier argument.
  Assume that $\Gamma$ is an unstable component of $\Sigma$. We will
  show that in this case the functions $f_{\tau_i}$ are bounded below
  $+\infty$ in a neighborhood of $\Gamma$. Hence $\Gamma$ lies in the
  interior of $\Omega^+\cup\Omega^0$ and can not be part of
  $\del\Omega^-$, which contradicts the assumption that $\Gamma$ is a
  component of $\Sigma$.

  If $\Gamma$ is unstable, let $\phi>0$ be a suitably scaled
  eigenfunction to the principal eigenvalue. We can extend the vector
  field $\phi\nu$ to a neighborhood of $\Gamma$, and flow $\Gamma$ by
  this vector field. This yields a map $F: \Gamma\times [-1,1] \to M$
  and constant $\Lambda>0$ with the following properties. We will
  denote $\Gamma_s = F(\Gamma,s)$.
\begin{enumerate}
\item $\Gamma_0 = \Gamma$.
\item $\Gamma_s \subset \Omega^+$ if $s\in[-1,0)$ and
  $\Gamma_s\cap\Omega^+=\emptyset$ if $s\in(0,1]$.
\item $\dd{F}{s} = \beta\nu$, where $\nu$ is the normal to $\Gamma_s$
  extending the outward pointing normal $\nu$ on $\Gamma$, and $\beta$
  satisfies the estimates
  \begin{equation*}
    \Lambda^{-1} \leq \beta \leq \Lambda,
    \qquad\text{and}\qquad
    \left|\dd{\beta}{s}\right|\leq \Lambda.  
  \end{equation*}
\item Outside of $\Omega^+$ we have $\theta^+[\Gamma_s] <0$ and inside
  $\theta^+[\Gamma_s]>0$ and
  \begin{equation*}
    \Lambda^{-1}s \leq |\theta^+[\Gamma_s]| \leq
    \Lambda s \qquad\text{for all}\qquad s\in[-1,1].
  \end{equation*}
\item
  We can assume that $\|K\|_{C^0(M)} \leq \Lambda$.
\end{enumerate}
For an interval $(s_1,s_2)\subset [-1,1]$ we denote by $A(s_1,s_2)$
the annular region $F\big(\Gamma\times(s_1,s_2)\big)$, which is
foliated by the $\Gamma_s$ for $s\in(s_1,s_2)$ and has boundary $\del
A(s_1,s_2) = \Gamma_{s_1} \cup \Gamma_{s_2}$.

We will construct a subsolution $w$ of Jang's equation, satisfying
$\CJ[w]\geq \eta>0$. The function $w$ will be constant on the
$\Gamma_s$, that is $w = \phi(s)$. We will later use the positivity of
$\eta$ to infer that $w+m_\tau$ are in fact subsolutions for
$\CJ_\tau$, where $m_\tau$ is a suitably chosen constant.
\begin{lemma}
  For $w = \phi(s)$ we can compute Jang's operator to be the following
  expression
  \begin{equation}
    \label{eq:11}
    \CJ[w] = \frac{\phi'}{\beta\sigma} \theta^+
    - \left( 1+ \frac{\phi'}{\beta\sigma}\right) P
    - \sigma^{-2} K(\nu,\nu)
    + \frac{\phi''}{\beta^{2}\sigma^{3}}
    -\frac{\phi'}{\beta^3\sigma^3}\dd{\beta}{s}.  
  \end{equation}
  Here $\sigma^2 = 1 + \beta^{-2}\phi'^2$.
\end{lemma}
To construct $w$ we will proceed in three steps, which amount to
constructing $w$ on the annuli $A_1:=A(-\delta,0)$, $A_2:=A(0,\eps)$,
and $A_3:=A(\eps,2\eps)$, where $\delta$ and $\eps$ will be fixed
during the construction.

We start with the construction of $\phi$ in $A_2 = A(0,\eps)$, which
will fix $\eps$, but not quite $\phi$. In this region all we know is
that $\theta^+[\Gamma_s]\leq 0$, so we we make the assumption
$\phi'\leq -\mu < 0$, where we will fix $\mu$ in the course of the
argument. This renders the first term in \eqref{eq:11} to be
non-negative. We can thus estimate that
\begin{equation}
  \label{eq:13}
  \CJ[w] \geq -\frac{c_1}{\mu^2} + c_2 \frac{\phi''}{|\phi'|^3},
\end{equation}
for constants $c_1,c_2>0$ depending only on $\Lambda$, provided we
choose $\mu\geq \Lambda$. To see this, note that $\sigma$ is
comparable to $|\phi'|$ provided the latter is bounded away from zero.
The fact that the term containing $P$ in \eqref{eq:11} is of the form
$c_1/\mu^2$ follows from the Taylor expansion of the square root. To
get that the right hand side of \eqref{eq:13} is positive we must
satisfy
\begin{equation}
  \label{eq:17}
  \frac{\phi_2''}{|\phi'|^3} \geq \frac{c_0}{\mu^2},
\end{equation}
where $c_0= \frac{c_1 + 1}{c_2} + 1$ is a positive constant depending only
on $\Lambda$. We will later use $c_0 > 1$ and $c_0 c_2 > 1$.

We make the following ansatz for $\phi$ in $[0,\eps]$:
\begin{equation}
  \label{eq:12}
  \phi_2(s) = a_2\left(1+\frac{s}{\eps}\right)^{2/3} + b_2
\end{equation}
for constants $a_2,b_2$ to be determined. We compute that
\begin{align}
  \label{eq:16}
  &
  \phi_2'(s) = \frac{2 a_2}{3\eps}\left(1+\frac{s}{\eps}\right)^{-1/3}
  \\
  \label{eq:21}
  &
  \phi_2''(s)
  =
  - \frac{2a_2}{9\eps^2} \left(1+\frac{s}{\eps}\right)^{-4/3}
  =
  - \frac{9\eps^2}{8a_2^3} \phi_2'(s)^4.
\end{align}
As we want to have $\phi_2'<0$, we must choose $a_2<0$ which renders
$\phi_2''(s)>0$. So in order to get $\phi'(s)\leq -\mu$ it is
sufficient to take
\begin{equation*}
  -\mu = \phi_2'(\eps) = \frac{a_2}{3\eps} 2^{2/3},
\end{equation*}
as $|\phi'|$ is increasing. This implies
\begin{equation}
  \label{eq:15}
  a_2^2 =  2^{-4/3}9\eps^2\mu^2.
\end{equation}
To satisfy \eqref{eq:17}, we require that
\begin{equation*}
  \frac{c_0}{\mu^2} \leq \frac{\phi''(\eps)}{|\phi'(\eps)|^3} =
  \frac{9\eps^2}{8a_2^3} \phi'(\eps) = \frac{3\eps}{a_2^2}
  2^{-7/3}.
\end{equation*}
This is equivalent to
\begin{equation}
  \label{eq:18}
  a_2^2 \leq \frac{3\eps \mu^2}{c_0} 2^{-7/3}.
\end{equation}
Combining with \eqref{eq:15} we find the condition
\begin{equation}
  \label{eq:19}
  9\eps^2\mu^2 2^{-4/3} \leq \frac{3\eps \mu^2}{c_0} 2^{-7/3}
\end{equation}
or
\begin{equation*}
  \eps \leq \frac{1}{6c_0}.
\end{equation*}
Thus we choose $\eps = \frac{1}{6c_0}$.  Note that since $c_0>1$,
$\eps<\frac{1}{6} <\frac{1}{2}$. Modulo fixing $\mu$ and the vertical
shift, we are done with $\phi$ on $(0,\eps)$. Note that $\eps$ does
not depend on $\mu$ which is important in view of the fact that we
will later choose $\mu$ as a function of $\eps$. Note further that
$\CJ[w]\geq \frac{1}{\mu^2}$ on $A_2$ by construction.

For $A_3 : = A(\eps,2\eps)$ we will make the ansatz $w = \phi_3(s)$,
with $s\in[\eps,2\eps)$. As we are in the region $s>\eps$, where
$\eps$ has been fixed by the construction in $A_2$, we have $\theta^+
\leq -\Lambda^{-1}\eps$ and thus the first term in \eqref{eq:11} is
estimated by $\kappa:=\frac{\eps}{\sqrt{2}\Lambda}>0$ from below. We
can estimate the whole expression as follows:
\begin{equation}
  \label{eq:29}
  \CJ[w] \geq \kappa - \frac{c_1}{\mu^2} - c_2 \frac{|\phi_3''(s)|}{|\phi_3'(s)|^3}
\end{equation}
where we again assumed $|\phi'(s)|\geq \mu\geq \Lambda$, and $c_1$
and $c_2$ are constants depending only on $\Lambda$. We can ensure
that the second term is small, that is
\begin{equation*}
  \frac{c_1}{\mu^2} \leq \frac{\kappa}{4}
\end{equation*}
provided
\begin{equation}
  \label{eq:20}
  \mu^2 \geq \frac{4c_1}{\kappa}.
\end{equation}
It remains to find a function,
which allows us to choose $\mu$ large while keeping the term
\begin{equation}
  \label{eq:22}
  c_2 \frac{|\phi_3''(s)|}{|\phi_3'(s)|^3} < \frac{\kappa}{4}.
\end{equation}
We make the ansatz
\begin{equation}
  \label{eq:23}
  \phi_3(s) = a_3 \log\left(1 - \frac{s-\eps}{\eps}\right) + b_3
\end{equation}
and compute
\begin{align*}
  \phi_3'(s) = -\frac{a_3}{\eps}\left(1 -
    \frac{s-\eps}{\eps}\right)^{-1},
  \\
  \phi_3''(s) = -\frac{a_3}{\eps^2}\left(1 -
    \frac{s-\eps}{\eps}\right)^{-2}. 
\end{align*}
As we need $\phi_3'(\eps)=-\mu$, to be able to fit $\phi_3$ to
$\phi_2$, we compute $-\mu = \phi_3'(\eps) = -\frac{a_3}{\eps}$ or
$a_3 = \eps\mu > 0$. Hence $\phi_3''(s)<0$ and $\phi_3'(s) \leq\mu$
for $s\in(\eps,2\eps)$, as desired. We still have to fix $\mu$. The
goal is to simultaneously satisfy \eqref{eq:20} and \eqref{eq:22}.
Compute
\begin{equation*}
  \frac{|\phi_3''(s)|}{|\phi_3'(s)|^3}
  =
  \frac{1}{\mu^2\eps}\left(1 - \frac{s-\eps}{\eps}\right)
  \leq
  \frac{1}{\mu^2\eps}.
\end{equation*}
Thus we can ensure \eqref{eq:22} provided $\mu^2 \geq
\frac{4c_2}{\eps\kappa}$.  We choose
\begin{equation*}
  \mu = \max \left\{ \sqrt{\frac{4c_1}{\kappa}},
    \sqrt{\frac{4c_2}{\eps\kappa}}, \Lambda \right\} 
\end{equation*}
and are done constructing $\phi_3$ up to fixing $b_3$ in such a way to
ensure $\phi_2(\eps) = \phi_3(\eps)$. Note that we have that
$\phi_3(s) \to -\infty$ as $s\to 2\eps$, which is the desired
behavior. Furthermore we have $\CJ[w]\geq \frac{\kappa}{2}>0$ in
$A_3$.

In the region $A_1 = A(-\delta,0)$, where $0<\delta<1$ will be chosen
later, we set $w(s) = \phi_1(s)$. Then we estimate from \eqref{eq:11}
that
\begin{equation}
  \label{eq:26}
  J[w] \geq -c_3 + c_4\frac{\phi_1''(s)}{|\phi_1'(s)|^3},  
\end{equation}
where $c_3$, and $c_4>0$ are again constants depending only on
$\Lambda$.  Here we assumed that $|\phi_1'(s)|\geq \Lambda$ as before.
The only chance to get the right hand side of this expression positive
is to take $\phi_1(s)$ to be a function with
\begin{equation*}
  \frac{\phi_1''(s)}{|\phi_1'(s)|^3} \geq \frac{c_3 + 1}{c_4} := c_5.
\end{equation*}
We make the ansatz
\begin{equation*}
  \phi_1(s) = a_1 \left(1 + \frac{s}{2\delta} \right)^{1/2} + b_1,
\end{equation*}
and compute
\begin{align*}
  \phi_1'(s)
  &=
  \frac{a_1}{4\delta} \left(1 + \frac{s}{2\delta}\right)^{-1/2}
  \\
  \phi_1''(s)
  &=
  -\frac{a_1}{16\delta^2} \left(1 + \frac{s}{2\delta}\right)^{-3/2}  
\end{align*}
We fix $b_1$ such that $\phi_1(-\delta) = 0$. This then fixes $b_2$
and $b_3$ by the requirement that $w$ is continuous on
$A(-\delta,2\eps)$. From the requirement $\phi_1'(0)=\phi_2'(0) =: -
\mu'$, we infer that
\begin{equation}
  \label{eq:27}
  a_1 = -4\mu'\delta.
\end{equation}
Recall that $-\mu'$ is fixed and can not be chosen freely. From
$\phi_1''(s)>0$ we find that $|\phi_1'(s)|\geq |\phi_1'(0)|=\mu' =
2^{1/3}\mu \geq\mu=|\phi_2(\eps)|\geq \Lambda$, so $\phi_1'$ is
automatically large enough to justify \eqref{eq:26}. To get positivity
of the right hand side of \eqref{eq:26} we need that
\begin{equation*}
  c_5 \leq \frac{\phi''(s)}{|\phi'(s)|^3} = \frac{4\delta}{a_1^2}.
\end{equation*}
Solving for $a_1^2$ yields the condition
\begin{equation}
  \label{eq:45}
  a_1^2 \leq \frac{4\delta}{c_5}.
\end{equation}
As we already fixed $a_1$ in \eqref{eq:27}, we infer the condition
\begin{equation*}
  \delta \leq \frac{1}{4c_5\mu'^2}.
\end{equation*}
So we fix $\delta = \frac{1}{4c_5\mu'^2}$ and are done. Note that
$\CJ[w]\geq 1$ by construction.

To summarize, we have constructed a function $w$ on $A(-\delta,2\eps)$
with the following properties:

\begin{mynum}
\item $w$ is $C^{1,1}$ up to the boundary in every $A(-\delta,s)$ with
  $s\in (-\delta,2\eps)$. Hence $w\in W^{2,\infty}\cap C^{1,1}$ away from
  $\Gamma_{2\eps}$,
\item there exists $\eta>0$ such that $J[w] \geq \eta$,
\item $w \equiv 0$ on $\Gamma_{-\delta}$, $w \leq 0$ on $A(-\delta,2\eps)$,
\item there exists $C_1 <\infty$ such that $0 \geq w \geq -C_1$ in
  $A(-\delta,\eps)$, and
\item $w|_{\Gamma_s} \to -\infty$ as $s\to 2\eps$.
\end{mynum}
Here $\eta$ and $C_1$ are constants that only depend on $\Lambda$, as
do $\delta$ and $\eps$.

With this subsolution $w$, we can get a lower bound of the functions
$f_\tau$ solving $\CJ[f_\tau] = \tau f_\tau$ near $\Gamma$ as
follows. Set
\begin{equation*}
  m:= \min \left\{ \inf_{\Gamma_{-\delta}} f_\tau , \frac{\eta}{\tau} \right\},
\end{equation*}
and consider the function
\begin{equation}
  \label{eq:55}
  w_m := w + m.
\end{equation}
The goal is to apply the comparison principle for the quasilinear
operator $\CJ$ to show that $w_m \leq f_\tau$ in $A(-\delta,2\eps)$.
To this end let $U$ be the region where $f_\tau \leq m$. From the
equation we conclude that
\begin{equation*}
  \CJ[f_\tau] =\tau f_\tau \leq \tau m \leq \eta
\end{equation*}
in $U$, and moreover $f_\tau = m$ on $\del U$. As $f_\tau \geq
-\frac{C}{\tau}$ is bounded below as in proposition~\ref{p:sup-est},
we can choose $\bar s\in(\eps,2\eps)$ such that $w_m|_{\Gamma_{\bar s}} \leq
\inf_M f_\tau$.

Set $V := U \cap A(-\delta,\bar s)$. Then, as $\del V \subset \del U
\cup \Gamma_{-\delta} \cup \Gamma_{\bar s}$, we find that $w_m \leq
f_\tau$ on $\del V$. An application of the comparison principle
\cite[Chapter 10]{Gilbarg-Trudinger:1998} allows
us to conclude that $w_m \leq f_\tau$ in $V$ and thus
\begin{equation*}
  w_m \leq f_\tau \quad \text{in}\quad A(-\delta, 2\eps).
\end{equation*}
By construction, there is a constant $C_1$ such that $w + C_1 \geq 0$
in $A(-\delta,\epsilon)$ and hence
\begin{equation*}
  m - C_1 \leq w_m \quad\text{in}\quad A(-\delta,\eps).
\end{equation*}
Thus we infer the estimate
\begin{equation}
  \label{eq:28}
  f_\tau \geq \min \left\{ \inf_{\Gamma_{-\delta}} f_\tau , \frac{\eta}{\tau} \right\}
  - C_1\quad\text{in}\quad A(-\delta,\eps).
\end{equation}
We can now conclude the argument. Take the sequence $\tau_i$ and the
functions $f_{\tau_i}$ from proposition~\ref{p:limit}. By construction
$f_{\tau_i}$ is uniformly bounded below on $\Gamma_{-\delta}$ as
$\Gamma_{-\delta}$ is compactly contained in $\Omega^+\cup\Omega^0$,
hence as $\tau_i\to 0$ the term on the right hand side of
\eqref{eq:28} is bounded below as $\tau_i\to 0$. Thus
$A(-\delta,\eps)\subset \Omega^+\cup\Omega^0$, which is a
contradiction, since we assumed that $\Gamma\subset A(-\delta,\eps)$
was a boundary component of $\del\Omega^-$.

This concludes the proof of theorem~\ref{thm:schoen-stable}.
\end{proof}

%%% Local Variables: 
%%% mode: latex
%%% TeX-master: "master"
%%% End: 

%% file: improved.tex
\section{Weak barriers}
\label{sec:constr-marg-trapp}
In this section we will slightly improve theorem~\ref{thm:schoen}
to allow interior boundaries where we just have the weak inequality
$\theta^+[\del^-M]\leq 0$, instead of the strict inequality assumed in
theorem~\ref{thm:schoen}.
\begin{theorem}
  \label{thm:schoen-weak}
  Let $(M,g,K)$ be a smooth, compact initial data set with $\del M$ the
  disjoint union $\del M = \del^-M \cup \del^+ M$ such that $\del^\pm
  M$ are non-empty, smooth, compact surfaces without boundary and
  $\theta^+[\del^-M]\leq 0$ with respect to the normal pointing into
  $M$ and $\theta^+[\del^+M]>0$ with respect to the normal pointing
  out of $M$.

  Then there exists a smooth, embedded, stable MOTS $\Sigma\subset M$
  homologous to $\del^+M$. $\Sigma$ may have components which agree
  with components of $\del^-M$ that satisfy $\theta^+ = 0$.
\end{theorem}
In this case we can not use the strong maximum principle to exclude
that $\Sigma$ touches $\del^-M$ as in
lemma~\ref{thm:boundary_as_barrier}.  For the proof of
theorem~\ref{thm:schoen-weak} we shall need the following lemma.
\begin{lemma}
  \label{lemma:flow}
  Let $\Sigma$ be a connected, two-sided, compact, embedded surface
  with $\theta^+ \leq 0$ and $\theta^+ \not\equiv 0$.  Then for every
  $\eps>0$ there exists a smooth, embedded surface $\Sigma'$ in the
  $\eps$-neighborhood of $\Sigma$, which lies to the outside of
  $\Sigma$ but does not touch $\Sigma$, is a graph over $\Sigma$, and
  satisfies $\theta^+ <0$.
\end{lemma}
\begin{proof}
  Consider the following equation for a function
  $F:\Sigma\times[0,\bar s)\to M$
  \begin{equation}
    \label{eq:5}
    \begin{cases}
      \DD{F}{s} = - \theta^+ \nu \\
      F(\cdot,0) = \id_{\Sigma}.     
    \end{cases}
  \end{equation}
  Here, $\nu$ is the outer normal as usual. This is a weakly parabolic
  equation for $F$, in fact it is a generalization of the mean
  curvature flow. To see this, recall that $\theta^+ = H + P$, where
  $H$ is the mean curvature, and $P = \trM K - K(\nu,\nu)$ is a term
  only depending on first derivatives of $F$. Thus the flow in
  equation \eqref{eq:5} is
  \begin{equation*}
    \DD{F}{s} = - H \nu - \text{lower order}
  \end{equation*}
  Hence it has the same symbol as the mean curvature flow and thus is
  a quasilinear parabolic equation.

  The theory of parabolic equations guarantees the existence of a
  solution for a small time interval $[0,\bar s)$, see for example
  \cite[Section 7]{huisken-polden:1999}.  Furthermore, any surface
  $\Sigma_s = F(\Sigma, s)$ for $s\in(0,\bar s)$ is smooth. From a
  standard argument using the strong maximum principle we conclude
  that $\theta^+ <0$ instantly. To see this, recall that the evolution
  equation for $\theta^+$ has the form
  \begin{equation*}
    \dd{\theta^+}{s}
    =
    - L_s \theta^+
    =
    \Delta \theta^+ - 2 S(\nabla \theta^+) - \theta^+ Q,
  \end{equation*}
  where $L_s$ is the linearization of $\theta^+$ along $\Sigma_s$,
  with
  \begin{equation*}
    Q = \divSig S -\tfrac{1}{2}|\chi^+|^2 - |S|^2 +
    \tfrac{1}{2}\ScalSig - \mu + J(\nu) - \half (\theta^+)^2 +
    \theta^+ \tr K,
  \end{equation*}
  where all geometric quantities are computed on $\Sigma_s$. Note that
  $L_s$ equals $L_M$ on MOTS. By smoothness we have that $Q$ is
  bounded for a short time, whence we can choose
  \begin{equation*}
    a> \max_{s\in [0,\bar s/2], x\in \Sigma_s} |Q(x,s)|.
  \end{equation*}
  Let $u=e^{-as} \theta^+$ and compute
  \begin{equation*}
    \left(\tfrac{\del}{\del s} - \Delta\right) u = -2S(\nabla u) - (Q+a) u.
  \end{equation*}
  The coefficient of the zeroth order term is negative. Hence the strong
  maximum principle from~\cite{Lieberman:1996} is applicable to $u$
  and implies that $u$ instantly becomes negative, implying that
  $\theta^+$ instantly becomes negative.

  If $s$ is small enough, $\Sigma_s$ will also be embedded. As
  $\theta^+\leq 0$, the flow \eqref{eq:5} moves the surface in the
  direction of $\nu$ everywhere, and hence outward, in particular
  $\Sigma_s\cap\Sigma =\emptyset$. As the initial speed is given by
  $|\theta^+|$, which is bounded, the surfaces $\Sigma_s$ will be
  arbitrarily close to $\Sigma$, as long as $s>0$ is small enough.
  Hence we can choose $\Sigma'$ to be one of the $\Sigma_s$.
\end{proof}
% \begin{remark}
%   \label{rem:flow}
%   This lemma makes sense even for surfaces which are not smooth. In
%   particular, it is sufficient to have a piecewise smooth surface
%   which is $C^{1,1}$, and where each of the smooth pieces satisfies
%   $\theta^+ \leq 0$ and $\theta^+<0$ at some point in the smooth part.
%   The resulting surface $\Sigma'$ will then be smooth.  This is due to
%   the fact that in these cases the weak version of $\theta^+$ is still
%   defined as $L^\infty$ function which satisfies the inequality
%   $\theta^+\leq 0$ in the sense of distributions.
% \end{remark}
\begin{proof}[Proof of theorem~\ref{thm:schoen-weak}]
  The main difficulty here is that $\del^-M$ may have multiple
  connected components $\del^-M = \Gamma_1 \cup \ldots \cup\Gamma_N$
  where some of the $\Gamma_k$ satisfy $\theta^+ = 0$, to which we can
  not apply lemma~\ref{lemma:flow} directly.
  
  Lemma~\ref{lemma:flow} allows us to flow the boundary components
  $\Gamma_k$ with $\theta^+\leq 0$ and $\theta^+\not\equiv 0$ in
  direction of their outer normal $\nu$, that is into $M$, to replace
  $M$ by a manifold $M_1$ which is such that $\del^-M_1$ is still
  embedded and each component of $\del^-M_1$ either has $\theta^+ <0$
  or $\theta^+=0$. As the boundary components with $\theta^+=0$ do not
  allow the application of theorem~\ref{thm:schoen}, we have to tweak
  them a little.

  Pick one such component $\Gamma$ of $\del^-M$ with
  $\theta^+[\Gamma]=0$, then there are three cases.  Either, as a
  MOTS, $\Gamma$ is not stable, $\Gamma$ is stable, but not strictly
  stable, or $\Gamma$ is strictly stable.

  When $\Gamma$ is not stable, let $\phi>0$ be an eigenfunction for
  the principal eigenvalue $\lambda <0$ for the operator $L_M$ on
  $\Gamma$. Extend the vector field $\phi\nu$ to a neighborhood of
  $\Gamma$ and flow $\Gamma$ for a short time interval along this
  vector field. This yields a foliation $\{\Gamma_s\}_{s\in[0,\eps)}$
  of a neighborhood of $\Gamma$, such that $\Gamma_0 = \Gamma$ and
  $\Gamma_s$ lies inside of $M$ and has $\theta^+<0$ when $s>0$.
  Hence, we push $\Gamma$ a little inward and obtain a strictly
  trapped surface.

  In the other two cases we need to flow the components with respect
  to the vector field $-\phi\nu$, where $\phi>0$ is again the
  principal eigenfunction of $L_M$ on $\Gamma$. So we have to assume
  that there is an extension $(M',g',K')$ of $(M,g,K)$ with $M\subset
  M'$, $g=g'|_M$ and $K=K'|_M$ such that $\del^-M$ lies in the
  interior of $M'$. Such an extension can be constructed by simply
  gluing $[0,1] \times \del^-M$ to $M$ along $\del^-M$ and smoothly
  extending $g$ and $K$ to the added piece. Keeping this in mind, we
  can now move the other boundary components $\Gamma$ inwards in the
  following way.

  If $\Gamma$ is strictly stable, then by flowing in direction
  $-\phi\nu$, we construct a foliation $\{\Gamma_s\}_{s\in(-\eps,0]}$
  of a neighborhood of $\Gamma$, such that $\Gamma_0 = \Gamma$ and
  $\Gamma_s$ lies in direction $-\nu$, that is outside of $M$ and has
  $\theta^+<0$ if $s<0$. We choose one of the $\Gamma_s$ as new inner
  boundary. We will later use the fact that the region between the
  former boundary $\Gamma$ and the new boundary $\Gamma_s$ is foliated
  by surfaces with $\theta^+<0$ to ensure that the constructed MOTS
  does not enter this region.

  The last case is where $\Gamma$ is stable but not strictly stable.
  In this case we also flow $\Gamma$ in direction $-\phi\nu$ and
  construct a foliation $\{\Gamma_s\}_{s\in(-\eps,0]}$ of a
  neighborhood of $\Gamma$, such that $\Gamma_0 = \Gamma$ and
  $\Gamma_s$ lies outside of $M$ and
  \begin{equation}
    \label{eq:7}
    \ddeval{}{s}{s=0}\theta^+[\Gamma_s] = 0.
  \end{equation}
  We will change the data $K'$ along the surfaces $\Gamma_s$ by replacing
  $K'$ by
  \begin{equation*}
    \tilde K = K' - \half \phi(s) h_s,
  \end{equation*}
  where $h$ is the metric on $\Gamma_s$ and $\phi:\IR\to\IR$ is a $C^1$
  function with $\phi(s) = 0$ for $s>0$. Note that
  $\tilde\theta^+[\Gamma_s]$, which means the quantity $\theta^+$
  computed with respect to the new data $(M',g',\tilde K)$, satisfies
  \begin{equation*}
    \tilde\theta^+[\Gamma_s] = \theta^+[\Gamma_s] -\phi(s).
  \end{equation*}
  As $\theta^+[\Gamma_s]$ vanishes to first order in $s$ at $s=0$ by
  \eqref{eq:7}, we can extend $\phi$ as a $C^{1,1}$ function to
  $\tilde M$ such that $\theta^+<0$ on \emph{all} $\Gamma_s$, if $s<0$
  is close enough to zero. Hence, this case is similar to the strictly
  stable case. It is clear that we can choose $\Gamma_s$ in such a way
  that $\|\tilde K\|_{C^1(\tilde M)} \leq 2 \|K\|_{C^1(M)}$.
  
  In summary, by this construction we have replaced $(M,g,K)$ by a
  manifold $(\tilde M,\tilde g, \tilde K)$ which are both embedded in
  a data set $(M',g',K')$. The outer boundaries of $M$ and $\tilde M$
  agree and have $\theta^+>0$, while the inner boundary of $\tilde M$
  has $\theta^+[\del^- \tilde M]<0$. The data $\tilde K$ is $C^{1,1}$
  in $\tilde M$.
  
  The set $U := M\setminus\tilde M \subset M'$, corresponding to the
  boundary components we moved out of $M$, is foliated by surfaces
  $\Sigma_s$ with $\theta^+[\Sigma_s] <0$ with respect to the data
  $(\tilde g, \tilde K)$.

  We can now invoke theorem~\ref{thm:schoen} to find a smooth,
  embedded, stable MOTS $\Sigma$ in $\tilde M$, which bounds with
  respect to $\del^-\tilde M$. Note that it is only necessary to
  assume $K\in C^{1,\alpha}$ for some $0<\alpha\leq 1$ for the theorem
  to apply. If one of the components $\Sigma'$ of $\Sigma$ enters $U$,
  say the component $U'$ of $U$, then let $\bar s := \min \{ s :
  \Sigma_s\cap\Sigma_k \neq\emptyset\}$, where the $\Sigma_s$
  constitute the foliation of $U'$ by strictly trapped surfaces, as
  above. At the point, where the minimum is assumed, the outward
  normals of $\Sigma'$ and $\Sigma_{\bar s}$ point into the same
  direction, and hence the strong maximum principle implies that
  $\Sigma_k=\Sigma_{\bar s}$, a contradiction.  Thus $\Sigma\cap U =
  \emptyset$, and $\Sigma\subset M$ is the desired solution.  Note
  that some components of $\Sigma$ might agree with
  components of $\del^-M$ which have $\theta^+=0$.
  
  The assertion that $\Sigma$ is stable then follow from
  theorem~\ref{thm:schoen-stable}.
\end{proof}
As an immediate consequence of theorem~\ref{thm:schoen-weak}, we infer
the following corollary.
\begin{corollary}
  \label{cor:outermost}
  Let $(M,g,K)$ be such that $\del M$ is the disjoint union $\del M =
  \del^- M \cup\del^+ M$, where $\del^+ M$ is non-empty with
  $\theta^+[\del^+M]>0$ and $\del^-M$ is possibly empty.  If $\Sigma$
  is an outermost MOTS homologous to $\del^+M$, then there do not
  exist outer trapped surfaces enclosing $\Sigma$. In particular,
  $\Sigma$ is a stable MOTS.
\end{corollary}
%%% Local Variables: 
%%% mode: latex
%%% TeX-master: "master"
%%% End: 

%% file: surgery.tex
\section{Surgery}
\label{sec:surgery}
In this section we describe a surgery procedure to construct an outer
trapped surfaces outside of a MOTS $\Sigma$ with small $i^+(\Sigma)$
and bounded curvature. In view of the existence part in
theorem~\ref{thm:schoen-weak}, we infer a lower bound on $i^+(\Sigma)$
for outermost MOTS. This implies an area estimate.

Moreover, the surgery procedure guarantees that a fixed amount of the
volume outside of $\Sigma$ is consumed. By iterating the surgery
procedure and application of theorem~\ref{thm:schoen-weak}, we then infer
that after a finite number of steps we arrive at a MOTS $\Sigma'$
outside of $\Sigma$ with a lower bound on $i^+(\Sigma')$.

As usual, we assume that $\Sigma$ is homologous to $\del^+ M$ and denote the
region bounded by $\Sigma$ and $\del^+M$, that is the outside of $\Sigma$, by
$\Omega$.
\subsection{Neck geometry}
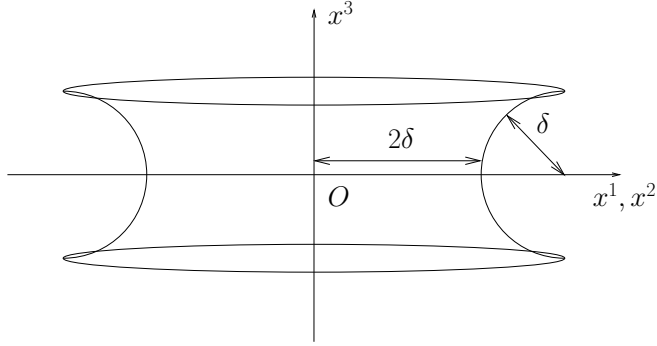
\begin{figure}[!h]
  \centering
  \resizebox{.6\linewidth}{!}{\input{pics/neck_pic.tex}}
  \caption{The $\delta$-standard neck.}
\label{fig:neck}
\end{figure}
The surgery procedure works by inserting necks with negative
$\theta^+$. We start by constructing a suitable neck in Euclidean
space, and transfer it to the geometry of $M$ in normal coordinates.
Let $\delta>0$ and consider the map
\begin{equation*}
  F : [0,2\pi] \times [-\tfrac{\pi}{2},\tfrac{\pi}{2}] \to \IR^3 :
  (\phi,\theta)
  \mapsto
  \begin{pmatrix}
    \delta\sin\phi(3-\cos\theta) \\
    \delta\cos\phi(3-\cos\theta) \\
    \delta\sin\theta
  \end{pmatrix}.
\end{equation*}
The image of $F$ is shown in figure~\ref{fig:neck}, we will call it
the $\delta$-standard neck. Denote by the interior $I_\delta$ of the
neck the points $(x^1,x^2,x^3)$ with $x^3\in(-\delta,\delta)$, $x^3 =
\delta\sin\theta$ and
\begin{equation*}
  (x^1)^2 + (x^2)^2 \leq  \delta^2(3-\cos\theta^2).
\end{equation*}
Clearly, the open ball $B^{\IR^3}_\delta(0)$ is contained in $I_\delta$.

The Euclidean mean curvature of the standard neck with respect to the
normal pointing out of $I_\delta$ is
\begin{equation*}
  H^e
  =
  -\delta^{-1} \big( 1 - (3-\cos\theta)^{-1}\cos\theta\big)
  \leq
  - (2\delta)^{-1}.  
\end{equation*}
Thus the Euclidean mean curvature of the $\delta$-standard neck can be
arbitrarily negative if $\delta$ is chosen small enough.
Let $r_0$ be such that at any point $O\in M$ with $\dist(O, \del
M)\geq \rho(M,g,K;\del M)/2$ we have geodesic normal coordinates
$\{x^i\}$ such that for $r\leq r_0$ we have
\begin{equation*}
  r^{-2} |g_{ij}-\delta_{ij}| + r^{-1} |\del_k g_{ij}| + |\del_k\del_l
  g_{ij}| \leq C
\end{equation*}
where $r$ is the Euclidean distance in $x$-coordinates. Then, the
image of the standard neck in these coordinates will have
$H<-(4\delta)^{-1}$ if $\delta<r_0$ is small enough. Thus, choosing
$\delta^{-1}$ large compared to $\|K\|_{C^0(M)}$, we can ensure that
the $\delta$-standard neck has $\theta^+ < 0$.

\subsection{Point selection}
The goal is to consume a fixed amount of volume by application of the
surgery. To this end, we have to insert a neck with $\delta$ bounded
away from zero in terms of the geometry of $M$. Hence, it is not
sufficient to do surgery at the points $p,q$ which realize
$i^+(\Sigma)$. Instead, we have to find points $p,q$ such that there
is a point $O$ with $\dist(O,\del M)\geq \rho(M,g,K;\del M)/2$ such
that $B^M_\delta(O)$ touches $\Sigma$ at $p$ and $q$, and the angle of
the segments joining $O$ to $p$ and $q$ at $O$ is close to $\pi$.
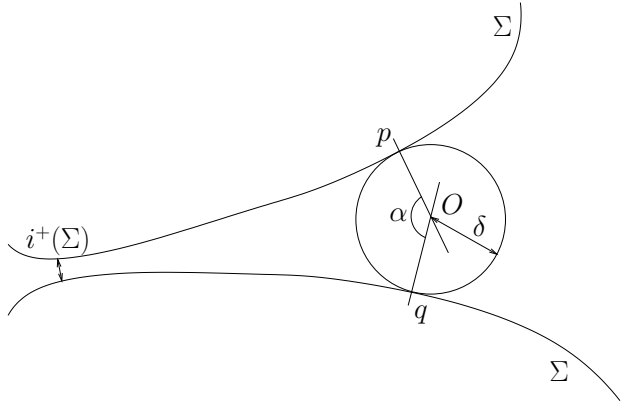
\begin{figure}[!b]
  \centering
  \resizebox{.6\linewidth}{!}{\input{pics/pointselect_pic.tex}}
  \caption{Selecting the points $p$ and $q$ where a ball $B_\delta(O)$
  touches $\Sigma$.}
\label{fig:pointselect}
\end{figure}

These points $p,q,O$ can be found as follows. Let $r_0$ be as
above. There exist $r_1<r_0$ and $C>0$ depending only on
$\|\RiemM\|_{C^0}$, such that $\del B^M_r(O)$ has second fundamental
form $A_r\geq \frac{C}{r} \gamma_r$ where $\gamma_r$ is the induced
metric on $\del B^M_r(O)$ (use the Hessian comparison theorem for the
distance function to $O$ \cite{schoen-yau:1994}). Furthermore, there
exists $0<r_2<r_1/2$, depending additionally on $\sup_\Sigma |A|$ with
the following property. If $O$ and $r<r_2$ are such that $\del
B^M_r(O)$ touches $\Sigma$ at $p$, then the $\Sigma$-ball
$B_{r_2}^\Sigma(p)$ does not intersect the interior of $B^M_r(O)$. The
important point to note is that the radius of the $\Sigma$-ball does
not depend on $r$.

Now fix $r<r_2$ and consider the set $U_r\subset\Sigma$ of points which
can be touched by a ball which lies completely outside of $\Sigma$,
that is, 
\begin{equation*}
  U_r
  :=
  \big\{ p \in \Sigma :
  \exists O\in \Omega\ \text{s.t.}\
  B^M_r(O) \subset\Omega\ \text{and}\ p \in \del B^M_r(O)
  \big\}.
\end{equation*}
Clearly $U_r$ is non-empty if $2r<\dist(\Sigma,\del^+ M)$, as then the
point $p_1\in\Sigma$ which realizes $\dist(\Sigma,\del^+ M)$ is in
$U_r$. Let $\Sigma_1$ be the component of $\Sigma$ containing $p_1$. If
$\Sigma_1\subset U_r$, then $\dist(\Sigma_1,
\Sigma\setminus\Sigma_1)\geq 2r$. We then select $p_2 \in
\Sigma\setminus\Sigma_1$ such that $p_2$ realizes the distance
$\dist(\Sigma\setminus\Sigma_1, \del^+M \cup\Sigma_1)$, clearly
$p_2\in U_r$. We can continue this process until either we found a
component $\Sigma_k$ of $\Sigma$ with $\Sigma_k\not\subset U_r$ and
$U_r\cap\Sigma_k \neq\emptyset$, or we showed that $\Sigma=U_r$.
However, the latter can not happen if $i^+(\Sigma) < r$, as the points
$p,q$ from lemma~\ref{lemma:almost_touch} are not in $U_r$. Thus,
there is a component $\Sigma_k$ of $\Sigma$ which contains a point $p
\in \del U_r$, the boundary of $U_r$ relative to $\Sigma$.

As $U_r$ is closed in $\Sigma$, there exists $O\in\Omega$ such that
$p\in \del B^M_r(O)$ and $B_r^M(O)\subset\Omega$. We claim that there
exists $q\in\Sigma\cap\del B^M_r(O)$, $q\neq p$. This $q$ can be
constructed as follows. Choose a sequence of points
$p_k\in\Sigma\setminus U_r$ with $p_k \to p$. Consider the geodesic
normal to $\Sigma$ emanating from $p_k$ outward. Let $O_k$ be the
point at distance $r$ from $p_k$ on this geodesic. As $p_k$ is not in
$U_r$, the ball $B_r(O_k)$ intersects $\Sigma$ in a point $q_k$ with
$\dist(q_k,O_k)< r$ and $\dist_\Sigma(p_k,q_k)\geq r_2$, by our choice
of $r$. By compactness we can assume that the $q_k$ converge to $q$
with $\dist(q,O) \leq r$ and $\dist_\Sigma(p,q)\geq r_2$. As $p\in
U_r$, the open ball $B^M_r(O)$ does not intersect $\Sigma$ and thus
$\dist(q,O)=r$.

Thus we find that, if $r<r_2$ and $i^+(\Sigma)<r$, there exist points
$p\neq q\in \Sigma$ and $O\in\Omega$ such that $p,q\in \del
B_r(O)$. Denote the geodesic segment joining $O$ and $p$ by $\gamma_p$
and the segment joining $O$ and $q$ by $\gamma_q$. We now want to show
that the angle between $\gamma_p$ and $\gamma_q$ at $O$ is close to
$\pi$ if $r$ is small enough.

Consider geodesic normal coordinates around $O$. Then the segments
$\gamma_p$ and $\gamma_q$ are straight lines emanating from $O$. let
$L_p$ be the plane orthogonal to $\gamma_p$ through $p$. As the
curvature of $\Sigma$ is bounded, $B^\Sigma_{r_3}(p)$ is the graph of
a function $u_p$ over $L_p$ with
\begin{equation}
  \label{eq:1}
  r^{-2} u_p + r^{-1} |\del_k u_p| + |\del_k\del_l u_p| \leq C
\end{equation}
for $r<r_3$ where $r_3>0$ and $C<\infty$ depend only on
$\inj_\rho(M,g,K;\del M)^{-1}$, $\|\RiemM\|_{C^0(M)}$ and $\sup_\Sigma
|A|)$. In particular, $B^\Sigma_{r_3}(p)$ is contained in a small
tubular neighborhood of $L_p$. Similarly, $B^\Sigma_{r_3}(q)$ is
contained in a neighborhood of $L_q$.

Let $\alpha$ be the angle of $\gamma_p$ and $\gamma_q$ at $O$.  We
claim that for each $\eta>0$ there exists $r>0$ such that
$|\alpha-\pi|<\eta$. Otherwise, if $\alpha$ is not close to
$\pi$, the planes $L_p$ and $L_q$ intersect at distance $d$ with $ d =
\frac{r}{\cos(\alpha/2)}\leq \frac{r}{\eps}$. Thus, choosing $r$ small
enough, we can make $L_p$ and $L_q$ intersect withing $d\leq
r_3/2$. This implies that $B^\Sigma_{r_3}(p)$ and $B^\Sigma_{r_3}(q)$
must also intersect. This is a contradiction, as $\Sigma$ is assumed
to be embedded.

\subsection{Surgery}
With the previous preparations, we can carry out the surgery
procedure. We choose $r$ so small that the above considerations apply,
giving the following properties.
\begin{enumerate}
\item\label{item:11}
  The $(2\delta)$-standard neck in normal coordinates around any point
  $O\in M$ with $\dist(O,\del M) > \inj_\rho(M,g,K;\del M)$ has
  $\theta^+ < 0$ in $(M,g,K)$.
\item\label{item:12}
  The $M$-ball $B_\delta^M(O)$ is contained in the interior of the
  image of the $(2\delta)$-standard neck.
\item\label{item:13}
  If $i^+(\Sigma)<\delta$, then there exist points $p,q\in\Sigma$
  and $O\in\Omega$ such that $B_\delta(O)\subset \Omega$ and
  $p,q\in\del B_\delta(O)$.
\item\label{item:14}
  The angle $\alpha$ of $\gamma_p$ and $\gamma_q$ at $O$
  satisfies $|1/\cos\alpha + 6\tan\alpha| \leq 3/2$. 
\end{enumerate}
Now assume that $i^+(\Sigma)<\delta$ and pick $p,q,O$ as in
condition~\ref{item:13} above, and consider geodesic normal
coordinates around $O$ such that $\gamma_q$ lies on the negative
$x^3$-axis. Let $N$ be the image of the $(2\delta)$-neck centered at
$O$ with its axis aligned with the $x^3$-coordinate axis, as in
figure~\ref{fig:surgery}.
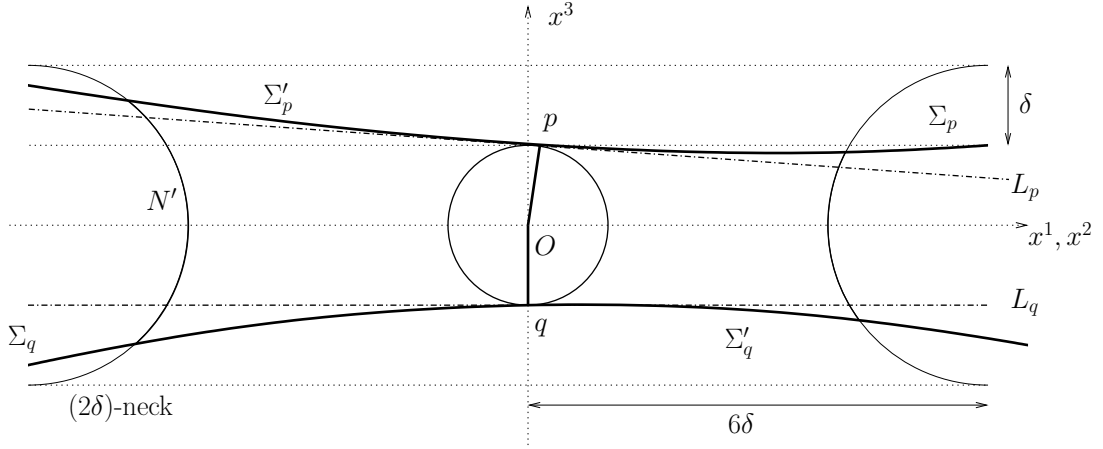
\begin{figure}[!t]
  \centering
  \resizebox{\linewidth}{!}{\input{pics/surgery_pic.tex}}
  \caption{The surgery in geodesic normal coordinates.}
\label{fig:surgery}
\end{figure}
Condition~\ref{item:14} on $\alpha$ implies that the plane $L_p$ is
such that 
\begin{equation*}
  L_p\cap\{(x^1)^2 + (x^2)^2 \leq 6\delta\}
  \subset
  \{ -\frac{3}{2}\delta \leq x^3 \leq \frac{3}{2}\delta]\}.
\end{equation*}
Recall that the component $\Sigma_p$ of $\Sigma\cap \{-2\delta\leq x^3 \leq 2\delta\}$
containing $p$ is the graph over $L_p$ of a function $u_p$ with
\begin{equation*}
  r^{-2} u_p + r^{-1} |\del_k u_p| + |\del_k\del_l u_p| \leq C ,
\end{equation*}
where $C$ is as in equation~\eqref{eq:1}. Thus, we can choose $\delta$,
depending only on $C$ so small, that first $\Sigma_p \subset
\{-2\delta\leq x^3 \leq 2\delta\}$, and second $\Sigma_p$ and $N$
intersect transversely (note that the angle of $\Sigma$ and $L_p$ is
of order $\delta$, whereas the angle between the neck and $L_p$ is
uniformly bounded away from zero).  We can similarly argue for $\Sigma_q$, so that we
find that figure~\ref{fig:surgery} is indeed accurate.

The surgery can now be performed as follows. Let $\Sigma_p'$ be the
component of $\Sigma\setminus N$ that contains $p$ and $\Sigma_q'$ be
the component that contains $q$. Let $N'$ be the component of
$N\setminus\Sigma$ between $\Sigma_p$ and $\Sigma_q$. Construct a
non-smooth surface $\Sigma_N$ by removing $\Sigma_p'$ and $\Sigma_q$
and adding $N'$. By construction this surface is homologous to
$\Sigma$, and hence to $\del^+M$. By condition~\ref{item:11}, we find
that the inserted neck has $\theta^+<0$. Condition~\ref{item:12}
implies that $B_\delta(O)$ is indeed contained in the neck we
added. Furthermore, at the corner $\Sigma\cap N'$, the normals $\nu_N$
of $N'$ and $\nu$ of $\Sigma$ enclose an angle $<\pi$.

We proceed by using lemma~\ref{thm:smooth_corner} to smooth out this
corner, thereby constructing a surface $\Sigma'$. This $\Sigma'$ lies
outside of $\Sigma_N$, and agrees with $\Sigma_N$ except in an
arbitrarily small neighborhood of the corner and has $\theta^+\leq 0$
and $\theta^+\not\equiv 0$. Note that in particular, the component of
$\Sigma'$, which contains part of $N'$ has $\theta^+<0$ somewhere.

\subsection{Results}
By the previous surgery procedure we arrive at the following
proposition
\begin{proposition}
  \label{thm:trapped_outside}
  Let $(M,g,K)$ be a data set such that $\del M$ is the disjoint union $\del M
  = \del^+ M \cup \del^- M$ of smooth compact surfaces without
  boundary. Assume that $\theta^+(\del^+ M) >0$ and if $\del^- M$ is
  non-empty, that $\theta^+(\del^-M)<0$. 
  
  There exists $\delta>0$ depending only on $\inj_\rho(M,g,K;\del
  M)^{-1}$, $\|\RiemM\|_{C^0}$ and $\|K\|_{C^1}$ with the following
  property. If $\Sigma\subset M$ is a stable MOTS, homologous to
  $\del^+ M$, bounding $\Omega$ together with $\del^+ M$, and
  $i^+(\Sigma) < \delta$, then there exists a MOTS $\Sigma'$ outside
  of $\Sigma$, homologous to $\del^+M$ and bounding $\Omega'$ together
  with $\del^+ M$ such that
  \begin{equation*}
    \Vol (\Omega') \leq Vol(\Omega) - v_0.
  \end{equation*}
  where $0 < v_0 := \inf \{ \Vol B^M_\delta(p) : \dist(p,\del M) \geq
  \delta \}$.
\end{proposition}
\begin{proof}
  The fact that $\Sigma$ is stable yields a curvature bound in view of
  theorem~\ref{thm:curv-est}. Then the above surgery procedure can be
  applied to construct $\Sigma'$.
\end{proof}
An immediate corollary of the above proposition is the following.
\begin{corollary}
  \label{thm:outermost_iplus}
  Let $(M,g,K)$ and $\delta$ be as in
  proposition~\ref{thm:trapped_outside}. If $\Sigma$ is an outermost
  MOTS in $M$, then $i^+(\Sigma)\geq \delta$.
\end{corollary}
\begin{proof}
  If $i^+(\Sigma)<\delta$, then proposition~\ref{thm:trapped_outside},
  guarantees the existence of a barrier surface outside of $\Sigma$,
  and theorem~\ref{thm:schoen-weak} implies the existence of a MOTS
  outside of $\Sigma$. Thus $\Sigma$ is not outermost.
\end{proof}
More importantly, as already indicated, the fact that a surgery takes
away a uniform amount of volume, gives a finiteness result, which
allows us to prove the following theorem.
\begin{theorem}
  \label{thm:uniform_trapped_outside}
  Let $(M,g,K)$ be a data set such that $\del M$ is the disjoint union $\del M
  = \del^+ M \cup \del^- M$ of smooth compact surfaces without
  boundary. Assume that $\theta^+(\del^+ M) >0$ and if $\del^- M$ is
  non-empty, that $\theta^+(\del^-M)<0$. Let $\delta$ be as in
  proposition~\ref{thm:trapped_outside}.

  If $\Sigma\subset M$ is a MOTS homologous to $\del^+ M$, then there
  exists a stable MOTS $\Sigma'$, with
  \begin{equation*}
    i^+(\Sigma') \geq \delta.
  \end{equation*}
  such that $\Sigma'$ does not intersect the region bounded by
  $\Sigma$ (and $\del^-M$ if non-empty).
\end{theorem}
\begin{proof}
  If $\Sigma$ is not stable we use theorem~\ref{thm:schoen-weak} with
  inner boundary $\Sigma$ to construct a stable MOTS $\Sigma_1$
  outside of $\Sigma$. If $i^+(\Sigma_1)<\delta$, then
  proposition~\ref{thm:trapped_outside} applies and yields a barrier
  outside of $\Sigma_1$ which can be fed into
  theorem~\ref{thm:schoen-weak} to construct a stable MOTS $\Sigma_2$
  outside of $\Sigma_1$. The region bounded by $\Sigma_1$ and
  $\Sigma_2$ has volume bounded below by $v_0$, where $v_0$ is from
  proposition~\ref{thm:trapped_outside}. If $i^+(\Sigma_2) <\delta$,
  we can iterate. As each step consumes at least volume $v_0$ outside
  of $\Sigma$, this procedure must end after a finite number of steps
  with a surface $\Sigma_k$ with $i^+(\Sigma_k)\geq \delta$.
\end{proof}
A lower bound on $i^+(\Sigma)$ can be used to estimate the area of
$\Sigma$. This area estimate is crucial to get the compactness of the
class of stable MOTS with $i^+(\Sigma)$ bounded below.
\begin{proposition}
  \label{prop:area-bound}
  Let $(M,g)$ be a compact Riemannian manifold with boundary, and
  $\Sigma\subset M$ an embedded, two-sided surface with bounded
  curvature $|A|\leq C$. Let
  \begin{equation*}
    \delta := \min \{ i_0^+(\Sigma), i^+(\Sigma)\}.
  \end{equation*}
  Then there exists an absolute constant $c$ such that the following
  area estimate holds:
  \begin{equation}
    |\Sigma| \leq  c (\delta^{-1} + \sup_\Sigma |A|) \Vol(M)
  \end{equation}
\end{proposition}
\begin{proof}
  Let $\nu$ be the outward pointing normal to the surfaces $\Sigma^s:=
  G_\Sigma(\Sigma,s)$ for $s\in[0,\delta]$, where $G_\Sigma$ is as in
  equation \eqref{eq:24}. Then $\divM(\nu) = H^s$, where $H^s$ denotes
  the mean curvature of $\Sigma^s$. As $\delta \leq i_0^+(\Sigma)$,
  the estimate
  \begin{equation*}
    |\divM \nu| \leq 2\sup_{\Sigma^s} |A| \leq 4 \sup_\Sigma |A|
  \end{equation*}
  follows from the definition of $i_0^+(\Sigma)$ (which has the bound
  on $\sup_{\Sigma^s} |A|$ built in).
  
  Let $\phi$ be a cut-off function with $\phi(s) = 1$ for $s\leq
  \delta/4$, $\phi=0$ for $s\geq \delta/2$ and $0\leq \phi'(s) \leq
  8\delta^{-1}$. Using the divergence theorem for the vector field $N=
  -\phi(s)\nu$ in the volume $U := G(\Sigma,[0,\delta))$, we infer that
  \begin{equation*}
    |\Sigma| = \int_\Sigma \la N,\nu \ra \dmu = \int_U \divM N \leq
    \Vol(U) |\div N|.
  \end{equation*}
  This yields the desired area estimate.
\end{proof}
As outermost MOTS are stable, and thus have bounded curvature, we can
combine this proposition with corollary~\ref{cor:outermost} to infer
the following area bound for outermost MOTS.
\begin{theorem}
  \label{thm:area-bound}
  Let $(M,g,K)$ be a smooth, compact initial data set with $\del M$
  the disjoint union $\del M = \del^-M \cup \del^+ M$, where $\del^+
  M$ is non-empty and has $\theta^+[\del^+ M]>0$, and $\theta^-[\del^-
  M] <0$ if $\del^-M$ is non-empty. Then, if $\Sigma$ is an
  outermost MOTS, we have the estimate
  \begin{equation*}
    |\Sigma| \leq C,
  \end{equation*}  
  where $C$ depends only on $\|\RiemM\|_{C^0(M)}$, $\|K\|_{C^1(M)}$,
  $\inj_\rho(M,g,K,\del M)^{-1}$, and $\Vol(M)$.
\end{theorem}
As the proof of the previous theorem does not assume that $\Sigma$ is
connected, it also implies a bound on the number of components of an
outermost MOTS.
\begin{corollary}
  \label{cor:comp-bound}
  Let $(M,g,K)$ as above. Then there exists a constant $N$, depending
  only on $\|\RiemM\|_{C^0(M)}$, $\|K\|_{C^1(M)}$,
  $\inj_\rho(M,g,K;\del M)^{-1}$, and $\Vol(M)$ such that any
  outermost MOTS has at most $N$ components.
\end{corollary}
\begin{proof}
  Since outermost MOTS are stable, their curvature is bounded in view
  of theorem~\ref{thm:curv-est}. This implies a lower bound on the area of
  each component. From theorem~\ref{thm:area-bound} we thus infer a
  bound on the number of components.
\end{proof}
%%% Local Variables: 
%%% mode: latex
%%% TeX-master: "master"
%%% End: 

%% file: pics/neck_pic.tex
\begin{picture}(0,0)%
\includegraphics{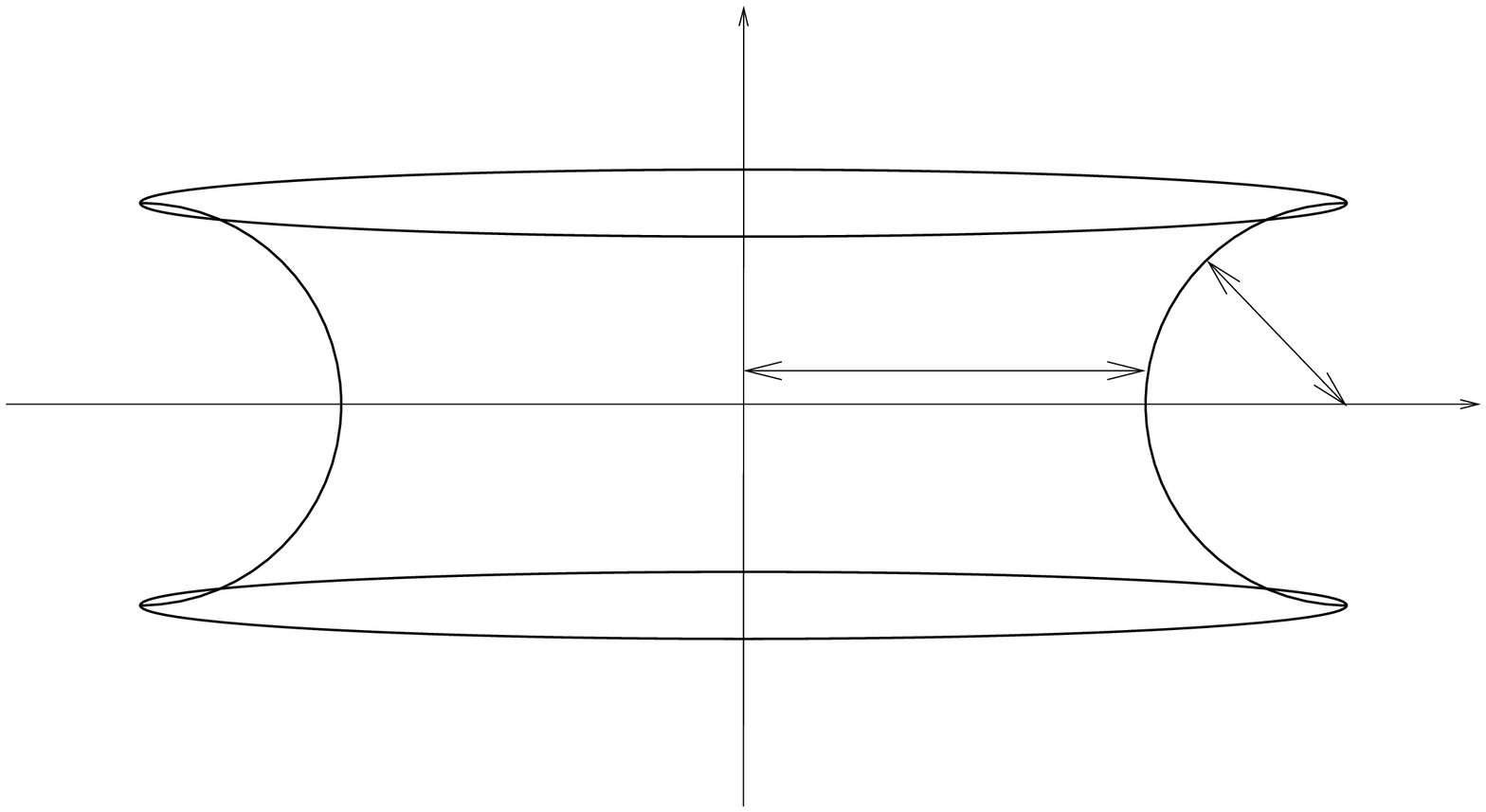}%
\end{picture}%
\setlength{\unitlength}{4144sp}%
\begingroup\makeatletter\ifx\SetFigFontNFSS\undefined%
\gdef\SetFigFontNFSS#1#2#3#4#5{%
  \reset@font\fontsize{#1}{#2pt}%
  \fontfamily{#3}\fontseries{#4}\fontshape{#5}%
  \selectfont}%
\fi\endgroup%
\begin{picture}(9924,5559)(3139,-5248)
\put(8326,-3050){\makebox(0,0)[lb]{\smash{{\SetFigFontNFSS{29}{34.8}{\familydefault}{\mddefault}{\updefault}{\color[rgb]{0,0,0}$O$}%
}}}}
\put(9300,-2100){\makebox(0,0)[lb]{\smash{{\SetFigFontNFSS{29}{34.8}{\familydefault}{\mddefault}{\updefault}{\color[rgb]{0,0,0}$2\delta$}%
}}}}
\put(11693,-1871){\makebox(0,0)[lb]{\smash{{\SetFigFontNFSS{29}{34.8}{\familydefault}{\mddefault}{\updefault}{\color[rgb]{0,0,0}$\delta$}%
}}}}
\put(8326,-100){\makebox(0,0)[lb]{\smash{{\SetFigFontNFSS{29}{34.8}{\familydefault}{\mddefault}{\updefault}{\color[rgb]{0,0,0}$x^3$}%
}}}}
\put(12600,-3050){\makebox(0,0)[lb]{\smash{{\SetFigFontNFSS{29}{34.8}{\familydefault}{\mddefault}{\updefault}{\color[rgb]{0,0,0}$x^1,x^2$}%
}}}}
\end{picture}%

%%% Local Variables: 
%%% mode: latex
%%% TeX-master: "master"
%%% End: 

%% file: pics/pointselect_pic.tex
\begin{picture}(0,0)%
\includegraphics{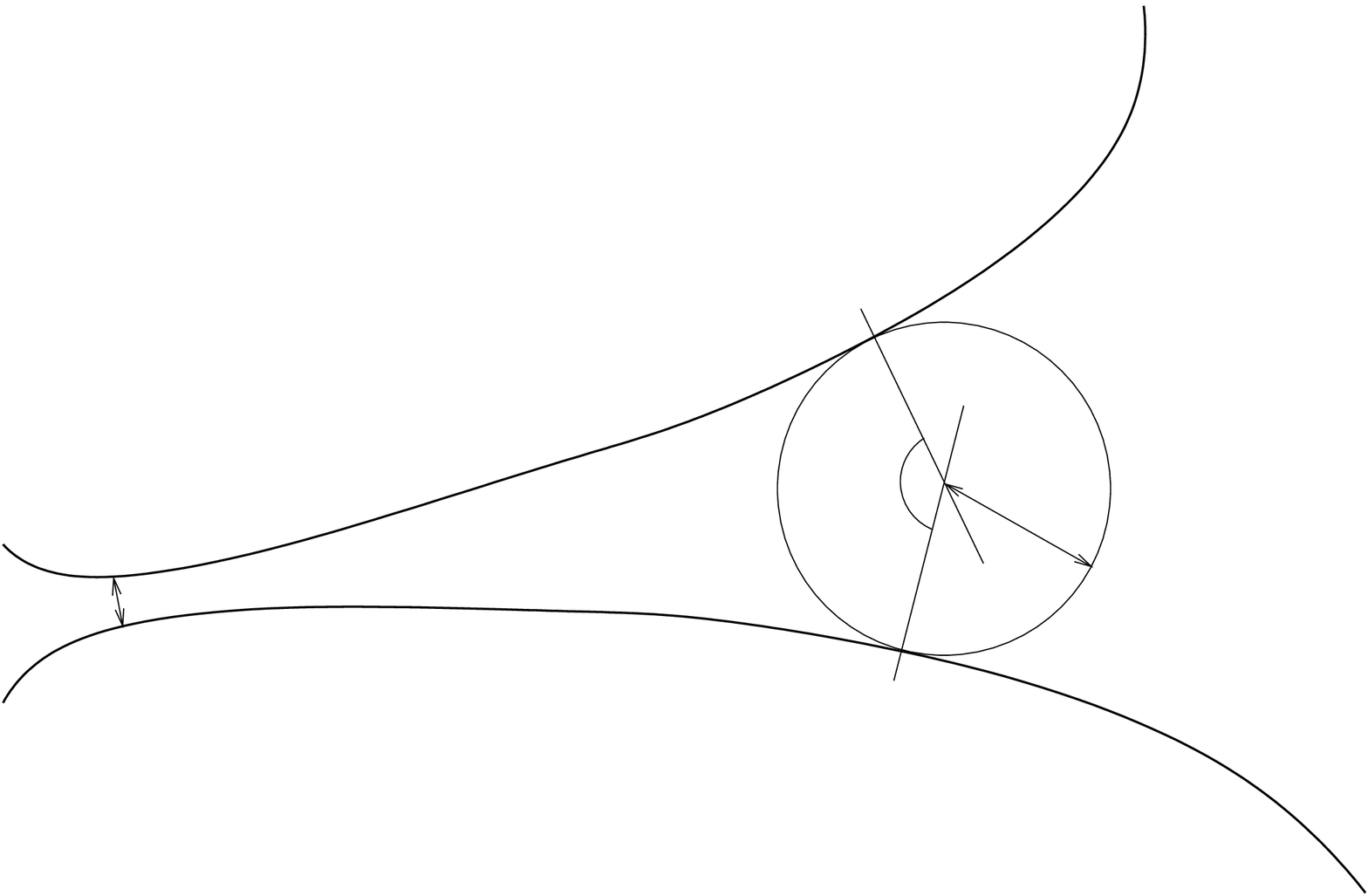}%
\end{picture}%
\setlength{\unitlength}{4144sp}%
\begingroup\makeatletter\ifx\SetFigFontNFSS\undefined%
\gdef\SetFigFontNFSS#1#2#3#4#5{%
  \reset@font\fontsize{#1}{#2pt}%
  \fontfamily{#3}\fontseries{#4}\fontshape{#5}%
  \selectfont}%
\fi\endgroup%
\begin{picture}(9719,6344)(3129,-7733)
\put(3433,-5300){\makebox(0,0)[lb]{\smash{{\SetFigFontNFSS{29}{34.8}{\familydefault}{\mddefault}{\updefault}{\color[rgb]{0,0,0}$i^+(\Sigma)$}%
}}}}
\put(10813,-1941){\makebox(0,0)[lb]{\smash{{\SetFigFontNFSS{29}{34.8}{\familydefault}{\mddefault}{\updefault}{\color[rgb]{0,0,0}$\Sigma$}%
}}}}
\put(11713,-7400){\makebox(0,0)[lb]{\smash{{\SetFigFontNFSS{29}{34.8}{\familydefault}{\mddefault}{\updefault}{\color[rgb]{0,0,0}$\Sigma$}%
}}}}
\put(8983,-3600){\makebox(0,0)[lb]{\smash{{\SetFigFontNFSS{29}{34.8}{\familydefault}{\mddefault}{\updefault}{\color[rgb]{0,0,0}$p$}%
}}}}
\put(9573,-6400){\makebox(0,0)[lb]{\smash{{\SetFigFontNFSS{29}{34.8}{\familydefault}{\mddefault}{\updefault}{\color[rgb]{0,0,0}$q$}%
}}}}
\put(10000,-4751){\makebox(0,0)[lb]{\smash{{\SetFigFontNFSS{29}{34.8}{\familydefault}{\mddefault}{\updefault}{\color[rgb]{0,0,0}$O$}%
}}}}
\put(10500,-5100){\makebox(0,0)[lb]{\smash{{\SetFigFontNFSS{29}{34.8}{\familydefault}{\mddefault}{\updefault}{\color[rgb]{0,0,0}$\delta$}%
}}}}
\put(9183,-4871){\makebox(0,0)[lb]{\smash{{\SetFigFontNFSS{29}{34.8}{\familydefault}{\mddefault}{\updefault}{\color[rgb]{0,0,0}$\alpha$}%
}}}}
\end{picture}%

%%% Local Variables: 
%%% mode: latex
%%% TeX-master: "master"
%%% End: 

%% file: pics/surgery_pic.tex
\begin{picture}(0,0)%
\includegraphics{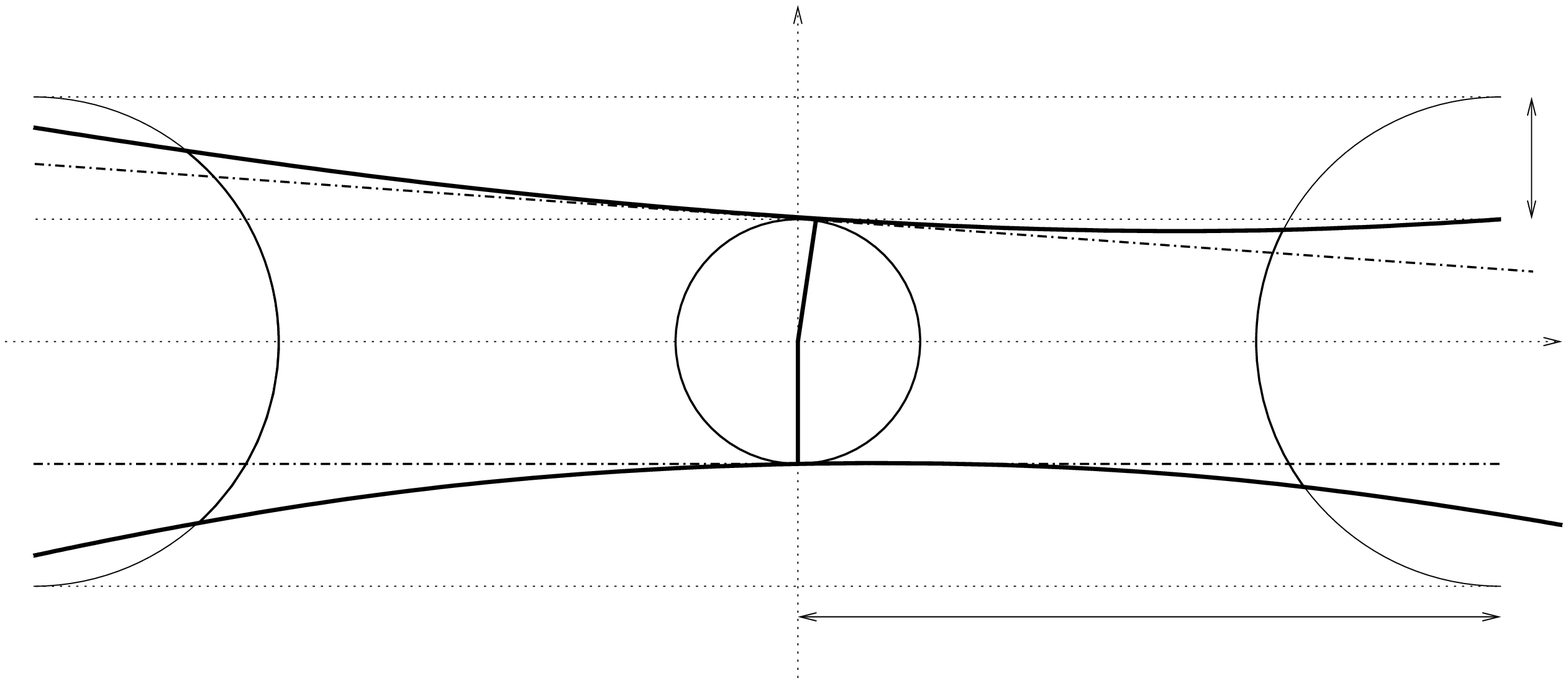}%
\end{picture}%
\setlength{\unitlength}{4144sp}%
\begingroup\makeatletter\ifx\SetFigFontNFSS\undefined%
\gdef\SetFigFontNFSS#1#2#3#4#5{%
  \reset@font\fontsize{#1}{#2pt}%
  \fontfamily{#3}\fontseries{#4}\fontshape{#5}%
  \selectfont}%
\fi\endgroup%
\begin{picture}(11523,4974)(1111,-4348)
\put(7050,-2200){\makebox(0,0)[lb]{\smash{{\SetFigFontNFSS{20}{20}{\familydefault}{\mddefault}{\updefault}{\color[rgb]{0,0,0}$O$}%
}}}}
\put(7050,-3056){\makebox(0,0)[lb]{\smash{{\SetFigFontNFSS{20}{20}{\familydefault}{\mddefault}{\updefault}{\color[rgb]{0,0,0}$q$}%
}}}}
\put(7150,-743){\makebox(0,0)[lb]{\smash{{\SetFigFontNFSS{20}{20}{\familydefault}{\mddefault}{\updefault}{\color[rgb]{0,0,0}$p$}%
}}}}
\put(12400,-1493){\makebox(0,0)[lb]{\smash{{\SetFigFontNFSS{20}{20}{\familydefault}{\mddefault}{\updefault}{\color[rgb]{0,0,0}$L_p$}%
}}}}
\put(12400,-2790){\makebox(0,0)[lb]{\smash{{\SetFigFontNFSS{20}{20}{\familydefault}{\mddefault}{\updefault}{\color[rgb]{0,0,0}$L_q$}%
}}}}
\put(1801,-4000){\makebox(0,0)[lb]{\smash{{\SetFigFontNFSS{20}{20}{\familydefault}{\mddefault}{\updefault}{\color[rgb]{0,0,0}$(2\delta)$-neck}%
}}}}
\put(9226,-4200){\makebox(0,0)[lb]{\smash{{\SetFigFontNFSS{20}{20}{\familydefault}{\mddefault}{\updefault}{\color[rgb]{0,0,0}$6\delta$}%
}}}}
\put(12500,-600){\makebox(0,0)[lb]{\smash{{\SetFigFontNFSS{20}{20}{\familydefault}{\mddefault}{\updefault}{\color[rgb]{0,0,0}$\delta$}%
}}}}
\put(7201,389){\makebox(0,0)[lb]{\smash{{\SetFigFontNFSS{20}{20}{\familydefault}{\mddefault}{\updefault}{\color[rgb]{0,0,0}$x^3$}%
}}}}
\put(12601,-2086){\makebox(0,0)[lb]{\smash{{\SetFigFontNFSS{20}{20}{\familydefault}{\mddefault}{\updefault}{\color[rgb]{0,0,0}$x^1,x^2$}%
}}}}
\put(11476,-736){\makebox(0,0)[lb]{\smash{{\SetFigFontNFSS{20}{20}{\familydefault}{\mddefault}{\updefault}{\color[rgb]{0,0,0}$\Sigma_p$}%
}}}}
\put(1126,-3211){\makebox(0,0)[lb]{\smash{{\SetFigFontNFSS{20}{20}{\familydefault}{\mddefault}{\updefault}{\color[rgb]{0,0,0}$\Sigma_q$}%
}}}}
\put(2680,-1685){\makebox(0,0)[lb]{\smash{{\SetFigFontNFSS{20}{20}{\familydefault}{\mddefault}{\updefault}{\color[rgb]{0,0,0}$N'$}%
}}}}
\put(9189,-3243){\makebox(0,0)[lb]{\smash{{\SetFigFontNFSS{20}{20}{\familydefault}{\mddefault}{\updefault}{\color[rgb]{0,0,0}$\Sigma_q'$}%
}}}}
\put(4000,-500){\makebox(0,0)[lb]{\smash{{\SetFigFontNFSS{20}{20}{\familydefault}{\mddefault}{\updefault}{\color[rgb]{0,0,0}$\Sigma_p'$}%
}}}}
\end{picture}%

%%% Local Variables: 
%%% mode: latex
%%% TeX-master: "master"
%%% End: 

%% file: trapped.tex
\section{The trapped region}
\label{sec:trapped-region}
In this section we examine the weakly outer trapped region, or more
precisely the boundary of the weakly outer trapped region. We make the
usual assumptions on $(M,g,K)$, that is $(M,g,K)$ is a smooth initial
data set with $\del M$ the disjoint union $\del M = \del^-M \cup
\del^+M$, where $\del^- M$ may be empty, but $\del^+ M$ is non-empty,
such that $\del^\pm M$ are, smooth, compact surfaces without boundary
and $\theta^+[\del^-M] < 0$ with respect to the normal pointing into
$M$ and $\theta^+[\del^+M]>0$ with respect to the normal pointing out
of $M$.

The definition of a trapped set and the trapped region below make
sense only if $\theta^+[\del^- M] < 0$. However, we can circumvent
this requirement for the main theorem as discussed in
remark~\ref{rem:weak-barrier} below.

To define the weakly outer trapped region, we introduce the notion of
a weakly outer trapped set.
\begin{definition}
  An open set $\Omega\subset M$ with smooth embedded boundary
  $\del\Omega$ is called \todef{weakly outer trapped set} if
  $\del\Omega$ is the disjoint union $\del\Omega = \del^-M \cup
  \del^+\Omega$ where $\del^+\Omega$ is a smooth, compact surface
  without boundary and $\theta^+[\del^+\Omega] \leq 0$ with respect to
  the normal pointing out of $\Omega$.
\end{definition}
Note that $\del^+\Omega$ is homologous to $\del^+ M$ in this
definition.
\begin{definition}
  The \todef{weakly outer trapped region} is the union of all weakly outer trapped
  sets enclosing $\del^-M$:
  \begin{equation}
    \label{eq:8}
    T := \bigcup_{\Omega\ \text{is outer trapped}} \Omega.
  \end{equation}
\end{definition}
We will henceforth refer to $T$ simply as the trapped region.  If
$\del^- M$ is non-empty, then the trapped region is non-empty as well,
but if $\del^-M$ is empty it might happen that $T$ is empty. In this
case the statements below are void.

Let $\del^-T : = \del T \cap \del^- M$ and $\del^+T = \del T \setminus
\del^-M$. The definition of $T$ is is analogous to the set
$\top_{\text{out},M}$ in \cite[Definition 3]{Kriele-Hayward:1997}. It
is known in the literature that provided $\del^+ T$ is smooth, it
satisfies $\theta^+=0$ \cite{Hawking-Ellis:1973,Kriele-Hayward:1997}.
The most general result about $\del^+ T$ we are aware of is
\cite[Proposition 7]{Kriele-Hayward:1997}, which asserts that if
$\del^+ T$ is $C^0$ and piecewise smooth, then it is smooth and
satisfies $\theta^+=0$. In contrast, we do not assume any initial
regularity for $\del^+ T$ for the following theorem.
\begin{theorem}
  \label{thm:trapped-region}
  Let $(M,g,K)$ be such that $\del M$ is the disjoint union $\del M =
  \del^+ M \cup \del^-M$ such that $\theta^+[\del^-M] <0$ if $\del^-M$
  is non-empty, and $\del^+ M$ is non-empty and has
  $\theta^+[\del^+M]>0$. Let $T$ be the trapped region in $M$. If $T$
  is non-empty, then $\del T$ is the disjoint union $\del T = \del^-T
  \cup \del^+T$ of smooth, compact surfaces without boundary, where
  $\del^-T = \del^-M$ and $\del^+T$ is a smooth stable MOTS homologous
  to $\del^+ M$.
\end{theorem}
\begin{remark}
  \label{rem:weak-barrier}
  If $(M',g',K')$ is a data set where $\del^-M'$ is only a weak
  barrier $\theta^+[\del^-M']\leq 0$, then $(M',g',K')$ can be
  modified to $(\tilde M, \tilde g, \tilde K)$ such that $\del^-\tilde
  M$ is a strong barrier $\theta^+[\del^-\tilde M]< 0$. This
  construction was already used in
  section~\ref{sec:constr-marg-trapp}. The trapped region $\tilde T
  \subset \tilde M$ of this extension is such that $\del^+\tilde T
  \subset M'$, that is, it lies in $M'$, since the region bounded by
  $\del^-\tilde M$ and $\del^- M'$ is a trapped set. However, it might
  be possible that $\del^+ \tilde T \cap \del^- M' \neq \emptyset$. In
  this case the intersection $\del^+ \tilde T \cap \del^- M'$ is a
  sub-collection of the components of $\del^-M'$ which are stable
  MOTS.
\end{remark}
\begin{remark}
  If the dominant energy condition holds, then $\del^+T$ is a
  collection of spheres or tori
  \cite{Hawking-Ellis:1973,Ashtekar-Krishnan:2003,Galloway-Schoen:2006}.
\end{remark}
The proof is along the lines of \cite[Section
4]{Huisken-Ilmanen:2001}. Before we begin the proof of the theorem
we prove some lemmas, which essentially replace the maximum
principle, which is not as powerful for MOTS, as it is for minimal
surfaces.
\begin{lemma}
  \label{lemma:transverse_intersect}
  Let $(M,g,K)$ be an initial data set as in
  theorem~\ref{thm:trapped-region} and $\delta>0$ be given. Let
  $\Omega_1\subset M$ and $\Omega_2\subset M$ be open sets such that
  $\del\Omega_i$ is the disjoint union $\del\Omega_i = \del^-M
  \cup\del^+\Omega_i$ where $\del^+\Omega_i$ is smooth, embedded, and
  $\del^+\Omega_i=\bigcup_{j=1}^{N_i}\Sigma_i^j$ is the union of
  disjoint, stable, connected MOTS $\Sigma_i^j$, $i=1,2$. Then for any
  $\delta>0$, there exists $\Omega_1'\subset \Omega_1$ and data $K'$
  on $M$ with the following properties:
  \begin{enumerate}
  \item \label{item:1} $\del\Omega_1' =\del^-M\cup\del^+\Omega_1'$.
  \item \label{item:2}
    $\del^+ \Omega_1'$ and $\del^+\Omega_2$ intersect transversally,
  \item \label{item:3}
    $\dist(\del^+\Omega_1',\del^+\Omega_1) <\delta$, 
  \item \label{item:4}
    $K'\in C^{1,1}(M)$ and $K'=K$ on $M\setminus\Omega_1$,
  \item \label{item:5}
    $\theta^+$ on $\del^+\Omega_2\cap M\setminus\Omega_1'$ computed
    with respect to $K'$ is at most its value with respect to $K$, and
  \item \label{item:6}
    there exists a foliation $\Sigma_s$, $s\in(-\eps,0]$ of
    $\Omega_1\setminus\Omega_1'$ such that $\Sigma_0=\del^+\Omega_1$ and
    $\theta^+[\Sigma_s]<0$ with respect to the data $K'$.
  \end{enumerate}
\end{lemma}
\begin{proof}
  By pushing the components of $\del^+\Omega_1$ into $\Omega_1$, as in
  the proof of theorem~\ref{thm:schoen-weak}, while changing the data
  $K$ to $K'$ near components of $\del\Omega_1$ which are stable but
  not strictly stable, we can construct $K'$ and a foliation
  $\Sigma_s$ near $\del\Omega_1$ such that each $\Sigma_s$ has
  $\theta^+[\Sigma_s]<0$, thus satisfying properties~\ref{item:1},
  \ref{item:4} and~\ref{item:6}. By Sard's theorem, $\Sigma_s$ and
  $\del^+\Omega_2$ intersect transversally for almost every
  $s\in(-\eps,0)$. Hence we can pick one such $s$, for which also
  properties~\ref{item:2} and~\ref{item:3} are satisfied.
  Property~\ref{item:5} follows by construction, as we were
  subtracting a non-negative definite tensor from $K$ to obtain $K'$.
\end{proof}
Subsequently, for two sets $\Omega_1,\Omega_2$ we denote by $\Omega_1
\triangle\Omega_2$ the symmetric difference, defined by $\Omega_1
\triangle\Omega_2 = ( \Omega_1 \setminus \Omega_2) \cup ( \Omega_2
\setminus \Omega_1)$.
\begin{lemma}
  \label{lemma:no_intersect}
  Let $(M,g,K)$, $\Omega_1$ and $\Omega_2$ be as in the previous
  lemma. Assume furthermore that $\Omega_1 \triangle
  \Omega_2\neq\emptyset$. Then there exists
  $\Omega\supset\Omega_1\cup\Omega_2$, such that $\del\Omega$ is the
  disjoint union $\del\Omega=\del^-M\cup\del^+\Omega$ where
  $\del^+\Omega$ is an embedded stable MOTS. Any connected component
  of $\del^+\Omega_1$ which intersects $\Omega_2$, lies in the
  interior of $\Omega$.
\end{lemma}
\begin{proof}
  There is nothing to prove if $\del (\Omega_1 \cup \Omega_2)$ is a
  smooth embedded manifold.  Thus we can assume that $\del^+\Omega_1$
  and $\del^+\Omega_2$ intersect.  Fix $\delta>0$ to be the distance
  at which we can apply proposition~\ref{thm:trapped_outside} in
  $(M,g,K)$. We use lemma~\ref{lemma:transverse_intersect}, to deform
  $\Omega_1$ and $K$ to $\Omega_1'$ and $K'$ with the stated
  properties for this choice of $\delta$. As $\del^+\Omega_1'$ and
  $\del^+\Omega_2$ intersect transversally,
  lemma~\ref{thm:smooth_corner} allows us to smooth out the corner of
  $\del (\Omega_1'\cup\Omega_2)$ in the outward direction.
  
  Furthermore, all stable components of $\del^+\Omega_1$ which were
  touching $\del\Omega_2$ but not intersecting $\Omega_2$ give rise to
  components of $\del^+\Omega_1'$, which are disjoint of
  $\del^+\Omega_2$ and at a distance at most $\delta$ to
  $\del^+\Omega_2$.  Thus we can apply the surgery procedure of
  proposition~\ref{thm:trapped_outside} to join these components to
  $\del\Omega_2$. This yields an open set $\Omega'$ with
  $\Omega'\supset\Omega_1'\cup\Omega_2$ and $\del\Omega'$ is the
  disjoint union $\del\Omega' = \del^-\Omega'\cup\del^+\Omega'$ where
  $\del^-\Omega'= \del^- M$ and $\del^+\Omega'$ is $C^{1,1}$ and has
  $\theta^+[\del^+\Omega']\leq 0$ and
  $\theta^+[\del^+\Omega']\not\equiv 0$, as $\theta^+\not\equiv 0$ on
  the components of $\del^+\Omega'$ which were created from joining a
  component of $\del^+\Omega_1'$ to a component of $\del^+\Omega_2$.
  We can then use the flow from lemma~\ref{lemma:flow} to smooth out
  the boundary of $\Omega'$, yielding
  $\Omega''\supset\Omega'\supset\Omega_1'\cup\Omega_2$ with $\Omega''$
  an open set. Note, by construction all components of
  $\del^+\Omega_1'$ and all components of $\del^+\Omega_2$ which were
  joined with components from $\del^+\Omega_1'$ are contained in the
  interior of $\Omega''$.

  Now an application of theorem~\ref{thm:schoen-weak} to the data
  $(M\setminus \Omega'',g,K)$, with inner boundary $\del^-(M\setminus
  \Omega'') = \del^+\Omega''$, and
  outer boundary $\del^+M$ yields a set $\Omega\supset\Omega''$ with
  boundary $\del\Omega$ the disjoint union
  $\del\Omega=\del^-M\cup\del^+\Omega$ where $\del^+\Omega$ is a smooth, stable MOTS. 

  By construction all components of $\del^+\Omega_1'$ and
  $\del^+\Omega_2$ are in the interior of $\Omega$.  Furthermore, an
  application of the strong maximum principle as in the proof of
  theorem~\ref{thm:schoen-weak} implies that $\del^+\Omega$ can not
  penetrate the region $\Omega_1\setminus\Omega_1'$ as this is
  foliated by trapped surfaces. In particular all components of
  $\del^+\Omega_1$ which meet $\del^+\Omega_2$ are contained in the
  interior of $\Omega$.
\end{proof}
\begin{remark}
  The preceding lemma implies the uniqueness of outermost MOTS.
\end{remark}
\begin{proof}[Proof of theorem~\ref{thm:trapped-region}]  
  Subsequently we assume that $T$ is non-empty, and therefore
  $(M,g,K)$ contains trapped regions, as otherwise there is nothing to
  prove.  We first show that we can define $\del^+ T$ by a collection
  of sets with much more well-behaved boundaries. We define $\CT$ to
  be the collection of all outer trapped sets $\Omega$, such that the
  outer boundary $\del^+\Omega$ satisfies the following four
  assumptions:
  \begin{enumerate}
  \item \label{item:7}
    $\theta^+[\del^+\Omega] = 0$,
  \item \label{item:8} every component of $\del^+\Omega$ is stable,
    and thus satisfies $\sup |A| \leq C$, where $C$ is the constant
    from theorem~\ref{thm:curv-est}, and depends only on
    $\|\RiemM\|_{C^0(M)}$, $\|K\|_{C^1(M)}$ and
    $\inj_\rho(M,g,K;\del M)$.
  \item \label{item:10} $i^+(\del^+\Omega) \geq \delta$ where $\delta$
    depending on the same data as $C$ above is the $\delta$ from
    theorem~\ref{thm:uniform_trapped_outside}.
  \item \label{item:9} $|\del^+\Omega| \leq C$, 
    where $C$ is the area resulting from
    proposition~\ref{prop:area-bound} applied to $\del^+\Omega$ with
    $i^+(\del^+\Omega)\geq\delta$ for the above $\delta$. This $C$
    also depends only on $\inj_\rho(M,g,K;\del M)$, $\|\RiemM\|_{C^0(M)}$ and
    $\|K\|_{C^1(M)}$.
  \end{enumerate}
  To this end, assume that $\Omega$ is an outer trapped set, which
  does not lie in $\CT$. Then we construct a set
  $\Omega'\supset\Omega$ which lies in $\CT$ by applying
  theorem~\ref{thm:uniform_trapped_outside} and using
  proposition~\ref{prop:area-bound} to prove the area estimate.
  
  We thus see that
  \begin{equation*}
    T = \bigcup_{\Omega\in\CT} \Omega.
  \end{equation*}
  The first claim is that for each point $p\in\del^+ T$ there exists
  $\Omega\in\CT$ such that $p\in \del^+ \Omega$. Clearly, for every
  $n$ there exists $\Omega_n$ such that $\dist(\Sigma_n, p) <
  \frac{1}{n}$, where $\Sigma_n = \del^+\Omega_n$. We can now appeal
  to the compactness theorem \cite[Theorem
  1.3]{Andersson-Metzger:2005} for stable MOTS with bounded curvature
  and bounded area, which, after passing to a sub-sequence, yields a
  limit $\Sigma$ of $\del^+\Omega_n$ in $C^{1,\alpha}$. This $\Sigma$
  is a smooth stable MOTS with bounded curvature and bounded area.
  Furthermore, $\Sigma$ is the outer boundary of a set $\Omega$, as
  the $\del^+\Omega_n$ can eventually be represented as graphs over
  $\Sigma$.
  
  However, $\Sigma$ is not necessarily embedded, as the limit of
  embedded surfaces might meet itself. As
  $i^+(\del^+\Omega)\geq\delta$, the only thing that prevents $\Sigma$
  from being embedded are points where $\Sigma$ touches itself from
  the inside. To remedy this, we can replace the sequence of the
  $\Omega_n$ by a sequence $\Omega_n'$ which is increasing in the
  sense that $\Omega_n'\subset \Omega_{n+1}'$ for all $n$. We proceed
  inductively and let $\Omega_1' := \Omega_1$. Assume that we have
  constructed
  \begin{equation*}
    \Omega_1' \subset \Omega_2' \subset \ldots \subset \Omega_{n-1}'
  \end{equation*}
  with $\Omega_k'\in CT$ for $k = 1, \ldots, n-1$.  Consider the set
  $\Omega_n \cup \Omega_{n-1}'$. Either this set has a smooth embedded
  boundary, in which case we can use theorem~\ref{thm:schoen-weak} to
  ensure the existence of $\Omega_n' \supset \Omega_n \cup
  \Omega_{n-1}'$ or $\Omega_n \cup \Omega_{n-1}'$ does not have a
  smooth boundary.  Then lemma~\ref{lemma:no_intersect} yields a
  barrier for theorem~\ref{thm:schoen-weak} and allows us to construct
  $\Omega_n' \supset \Omega_n \cup \Omega_{n-1}'$. By eventually applying
  theorem~\ref{thm:uniform_trapped_outside}, we can assume that
  $\Omega_n'\in\CT$.

  We will now relabel $\Omega_n := \Omega_n'$ and
  $\Sigma_n:=\Sigma_n'$. As explained above, there is a subsequence of
  the $\Omega_n$ such that the $\Sigma_n$ converge in $C^{1,\alpha}$
  to a stable MOTS $\Sigma$ which is the outer boundary of a set
  $\Omega$ and has $i^+(\Sigma)\geq\delta$, thus $\Sigma$ can not
  touch itself on the outside. Since the $\Omega_n$ are increasing,
  $\Sigma$ can not touch itself on the inside either. This follows
  from the fact that the $\Sigma_n$ converge as graphs from the inside
  to $\Sigma$. Thus if $\Sigma$ touches itself on the inside, so would
  the $\Sigma_n$. But each $\Sigma_n$ is embedded, and hence $\Sigma$
  is embedded and $\Omega\in\CT$.
  
  Next we show that $\del^+ T$ consists of a smooth collection of
  MOTS. To this end assume first that $\Omega_1$ and $\Omega_2$ are
  such that the outer boundaries $\del^+\Omega_k$ meet $\del^+ T$ for
  $k=1,2$. Let $\Sigma_k$ be a component of $\del^+\Omega_k$ that
  meets $\del T$. From lemma~\ref{lemma:no_intersect} we infer that
  either $\Sigma_1=\Sigma_2$ or $\dist (\Sigma_1, \Sigma_2) > 0$.

  Furthermore, the surfaces which meet $\del^+ T$ have positive
  distance to each other. To this end suppose that $\Omega$ and
  $\Omega_i$, are such that a component $\Sigma$ of $\del^+\Omega$
  meets $\del T$ and a component $\Sigma_i$ of $\del^+\Omega_i$ which
  meets $\del^+T$, $i\geq 1$ is such that $\Sigma \cap \Sigma_i =
  \emptyset$, but the closure of $\bigcup_{i\geq 1} \Sigma_i$
  intersects $\Sigma$. Note that we cannot have $\Sigma_i\subset
  \Omega$ as by the strong maximum principle, $\Sigma_i$ would have to
  agree with $\Sigma$ in that case. Hence
  $\dist(\Sigma_i,\del\Omega)\to 0$ as $i\to\infty$ and $\Sigma_i$
  lies outside of $\Omega$. In this case we can eventually apply the
  surgery procedure of proposition~\ref{thm:trapped_outside} to join
  $\Sigma_i$ to $\Sigma$, which yields an open set $\Omega'\in\CT$
  such that $\Sigma_i\cup\Sigma$ is contained in the interior of
  $\Omega'$, which contradicts the fact that $\Sigma$ and $\Sigma_i$
  meet $\del T$.

It follows that $\del T$ is a collection of disjoint stable MOTS.
\end{proof}
%%% Local Variables: 
%%% mode: latex
%%% TeX-master: "master"
%%% End: 

%% file: ack.tex
\section*{Acknowledgements}
The authors wish to thank Walter Simon, Marc Mars, Greg Galloway, Rick
Schoen and Gerhard Huisken for helpful conversations. The second
author would also like to thank Michael Eichmair and Leon Simon for
their comments.

%%% Local Variables: 
%%% mode: latex
%%% TeX-master: "master"
%%% End: 

%% file: master.bbl
\newcommand{\etalchar}[1]{$^{#1}$}
\providecommand{\bysame}{\leavevmode\hbox to3em{\hrulefill}\thinspace}
\providecommand{\MR}{\relax\ifhmode\unskip\space\fi MR }
\providecommand{\MRhref}[2]{%
  \href{http://www.ams.org/mathscinet-getitem?mr=#1}{#2}
}
\providecommand{\href}[2]{#2}
\begin{thebibliography}{CLZ{\etalchar{+}}07}

\bibitem[AG05]{Ashtekar-Galloway:2005}
A.~Ashtekar and G.~J. Galloway, \emph{Some uniqueness results for dynamical
  horizons}, Adv. Theor. Math. Phys. \textbf{9} (2005), no.~1, 1--30.
  \MR{2193368 (2006k:83101)}

\bibitem[AK03]{Ashtekar-Krishnan:2003}
A.~Ashtekar and B.~Krishnan, \emph{Dynamical horizons and their properties},
  Phys. Rev. D (3) \textbf{68} (2003), no.~10, 104030, 25.

\bibitem[AM05]{Andersson-Metzger:2005}
L.~Andersson and J.~Metzger, \emph{Curvature estimates for stable marginally
  trapped surfaces}, arXiv:gr-qc/0512106, 2005.

\bibitem[AMS05]{Andersson-Mars-Simon:2005}
L.~Andersson, M.~Mars, and W.~Simon, \emph{Local existence of dynamical and
  trapping horizons}, Phys. Rev. Lett. \textbf{95} (2005), 111102,
  arXiv:gr-qc/0506013.

\bibitem[AMS07]{andersson-mars-simon:2007}
\bysame, \emph{Stability of marginally outer trapped surfaces and existence of
  marginally outer trapped tubes}, arXiv:0704.2889 [gr-qc], 2007.

\bibitem[BD04]{Ben-Dov:2004}
I.~Ben-Dov, \emph{Penrose inequality and apparent horizons}, Phys. Rev. D (3)
  \textbf{70} (2004), no.~12, 124031, 11.

\bibitem[CLZ{\etalchar{+}}07]{Campanelli:2006fy}
M.~Campanelli, C.~O. Lousto, Y.~Zlochower, B.~Krishnan, and D.~Merritt,
  \emph{Spin flips and precession in black-hole-binary mergers}, Phys. Rev.
  \textbf{D75} (2007), 064030.

\bibitem[CM99]{Colding-Minicozzi:1999}
T.~H. Colding and W.~P. Minicozzi, II, \emph{Examples of embedded minimal tori
  without area bounds}, Internat. Math. Res. Notices (1999), no.~20,
  1097--1100.

\bibitem[Dea03]{Dean:2003}
B.~Dean, \emph{Compact embedded minimal surfaces of positive genus without area
  bounds}, Geom. Dedicata \textbf{102} (2003), 45--52.

\bibitem[Eic]{Eichmair:2007}
M.~Eichmair, \emph{Thesis, in preparation}, Stanford University.

\bibitem[GS06]{Galloway-Schoen:2006}
G.~J. Galloway and R.~Schoen, \emph{A generalization of {H}awking's black hole
  topology theorem to higher dimensions}, Comm. Math. Phys. \textbf{266}
  (2006), no.~2, 571--576. \MR{2238889}

\bibitem[GT98]{Gilbarg-Trudinger:1998}
D.~Gilbarg and N.~S. Trudinger, \emph{Elliptic partial differential equations
  of second order}, {R}ev. 3. printing. {S}econd ed., Springer-Verlag Berlin,
  Heidelberg, New York, 1998.

\bibitem[HE73]{Hawking-Ellis:1973}
S.~W. Hawking and G.~F.~R. Ellis, \emph{The large scale structure of
  space-time}, Cambridge University Press, London, 1973, Cambridge Monographs
  on Mathematical Physics, No. 1.

\bibitem[HI01]{Huisken-Ilmanen:2001}
G.~Huisken and T.~Ilmanen, \emph{The inverse mean curvature flow and the
  {R}iemannian {P}enrose inequality}, J. Differential Geom. \textbf{59} (2001),
  no.~3, 353--437.

\bibitem[HP99]{huisken-polden:1999}
G.~Huisken and A.~Polden, \emph{Geometric evolution equations for
  hypersurfaces}, Calculus of variations and geometric evolution problems
  (Cetraro, 1996), Lecture Notes in Math., vol. 1713, Springer, Berlin, 1999,
  pp.~45--84. \MR{1731639 (2000j:53090)}

\bibitem[Jan78]{Jang:1978}
P.~S. Jang, \emph{On the positivity of energy in general relativity}, J. Math.
  Phys. \textbf{19} (1978), 1152--1155.

\bibitem[KH97]{Kriele-Hayward:1997}
M.~Kriele and S.~A. Hayward, \emph{Outer trapped surfaces and their apparent
  horizon}, J. Math. Phys. \textbf{38} (1997), no.~3, 1593--1604.

\bibitem[KLZ07]{Krishnan:2007pu}
B.~Krishnan, C.~O. Lousto, and Y.~Zlochower, \emph{Quasi-local linear momentum
  in black-hole binaries}, arXiv:0707.0876 [gr-qc], 2007.

\bibitem[Lie96]{Lieberman:1996}
Gary~M. Lieberman, \emph{Second order parabolic differential equations}, World
  Scientific Publishing Co. Inc., River Edge, NJ, 1996. \MR{MR1465184
  (98k:35003)}

\bibitem[NR06]{nabutowsky:rotman:2006}
A.~Nabutovsky and R.~Rotman, \emph{Curvature-free upper bounds for the smallest
  area of a minimal surface}, Geom. Funct. Anal. \textbf{16} (2006), no.~2,
  453--475.

\bibitem[Sch04]{Schoen:2004}
R.~Schoen, {Talk given at the Miami Waves conference}, January 2004.

\bibitem[SSY75]{Schoen-Simon-Yau:1975}
R.~Schoen, L.~Simon, and S.~T. Yau, \emph{Curvature estimates for minimal
  hypersurfaces}, Acta Math. \textbf{134} (1975), no.~3-4, 275--288.

\bibitem[SY81]{Schoen-Yau:1981}
R.~Schoen and S.-T. Yau, \emph{Proof of the positive mass theorem. {II}}, Comm.
  Math. Phys. \textbf{79} (1981), no.~2, 231--260.

\bibitem[SY94]{schoen-yau:1994}
R.~Schoen and S.-T. Yau, \emph{Lectures on differential geometry}, Conference
  Proceedings and Lecture Notes in Geometry and Topology, International Press,
  Boston, 1994.

\bibitem[Yau01]{Yau:2001}
S.-T. Yau, \emph{Geometry of three manifolds and existence of black hole due to
  boundary effect}, Adv. Theor. Math. Phys. \textbf{5} (2001), no.~4, 755--767.

\end{thebibliography}
